\documentclass[11pt, a4paper]{article}

\usepackage[utf8]{inputenc}
\usepackage[T1]{fontenc}
\usepackage{lmodern}

\usepackage[margin=1in]{geometry}
\usepackage{setspace}
\usepackage{titlesec}
\titleformat*{\section}{\large\bfseries\boldmath}
\titleformat*{\subsection}{\normalsize\bfseries\boldmath}
\titleformat*{\paragraph}{\normalsize\bfseries\boldmath}
\usepackage{amsmath, amssymb, amsthm, amsfonts, mathtools, enumitem} 
\allowdisplaybreaks

\usepackage[colorinlistoftodos, textsize=footnotesize]{todonotes}
\usepackage[authoryear,round]{natbib}

\definecolor{darkblue}{rgb}{0.0, 0.0, 0.55}
\usepackage[pagebackref=true, colorlinks=true, allcolors=darkblue]{hyperref}
\renewcommand*{\backref}[1]{}
\renewcommand*{\backrefalt}[4]{%
    \ifcase #1%
    \or #2%
    \else #2%
    \fi%
}

\makeatletter
\newtheoremstyle{smallcaps}
  {}% Space above
  {}% Space below
  {\itshape}% Body font (Italic)
  {}% Indent amount
  {\scshape}% Theorem head font (Small Caps)
  {.}% Punctuation after theorem head
  { }% Space after theorem head
  {}% Theorem head spec

\renewenvironment{proof}[1][\proofname]{\par
  \pushQED{\qed}%
  \normalfont \topsep6\p@\@plus6\p@\relax
  \trivlist
  \item[\hskip\labelsep
        \scshape
    #1\@addpunct{.}]\ignorespaces
}{%
  \popQED\endtrivlist\@endpefalse
}
\makeatother
\theoremstyle{smallcaps}
\newtheorem{theorem}{Theorem}[section]
\newtheorem{lemma}[theorem]{Lemma}

\newtheorem{claim}[theorem]{Claim}

\newcommand{\exsub}[2]{\mathbb{E}_{#1}\left[#2\right]}
\newcommand{\exsubpar}[2]{\mathbb{E}_{#1}[#2]}

\newcommand{\eps}{\epsilon}
\newcommand{\bbR}{\mathbb{R}}
\newcommand{\bbZ}{\mathbb{Z}}
\newcommand{\bbN}{\mathbb{N}}
\newcommand{\opt}{\mathrm{OPT}}
\newcommand{\myround}{\mathcal{R}}
\newcommand{\mylcm}{\mathrm{LCM}}

\newcommand{\mysecond}{\mathrm{sec}}
\newcommand{\myprim}{\mathrm{prim}}
\newcommand{\mydisc}{\mathrm{disc}}

\title{\boldmath \textbf{Resource-Constrained Joint Replenishment \\
via Power-of-$m^{1/k}$ Policies}}
\author{
    Danny Segev\thanks{School of Mathematical Sciences and Coller School of Management, Tel Aviv University, Tel Aviv 69978, Israel. Email: \texttt{segevdanny@tauex.tau.ac.il}.}
}
\date{}

\begin{document}

\maketitle

\begin{abstract}
\noindent The continuous-time joint replenishment problem has long served as a foundational inventory management model. Even though its unconstrained setting has seen recent algorithmic advances, the incorporation of resource constraints into this domain precludes the application of newly discovered synchronization techniques. Such constraints arise in a broad spectrum of practical environments where resource consumption is bounded as an aggregate rate over time. Prominent examples include inbound logistics facing volumetric limits on shipping containers, in-plant manufacturing restricted by material handling bandwidth, and administrative policies capping cumulative purchasing budgets. However, for nearly four decades, the prevailing approximation guarantee for resource-constrained joint replenishment has remained $\frac{ 1 }{ \ln 2 } \approx 1.4427$, achieved via classical power-of-$2$ policies.

\medskip \noindent In this paper, we circumvent these structural policy restrictions by devising generalized rounding frameworks, demonstrating that a well-known convex relaxation is much tighter than previously established. In particular, we expand our analytical scope to encompass fractional base expansion factors, randomized shifting, and staggered interleaved grids. Through this multifaceted methodology, we present a sequence of gradually improving performance guarantees. First, by proposing a best-of-two framework that exploits structural asymmetries between deterministic power-of-$m^{1/k}$ policies, we surpass the classical barrier to obtain a $1.3776$-approximation. Second, by injecting a random shift into the logarithmic grid domain and formulating a factor-revealing linear program to optimize a dual-policy approach, we attain a $1.2512$-approximation. Finally, by superimposing a secondary offset grid to subdivide rounding intervals and suppress holding cost inflation, we utilize interleaved policies to arrive at our ultimate approximation ratio of $\frac{5}{6\ln 2} \approx 1.2023$, which is proven to be best-possible for the class of interleaved power-of-$m^{1/k}$ policies.
\end{abstract}

\vspace{1em}

\thispagestyle{empty}

% --- Table of Contents & List of Todos ---
\newpage
\thispagestyle{empty}
\tableofcontents
% \newpage
% \thispagestyle{empty}
% \listoftodos[List of Comments]

% --- Start of Content ---
\newpage
\setcounter{page}{1}
\section{Introduction}

Since its formal inception in the mid-1960s, with indirect explorations surfacing even earlier, the continuous-time joint replenishment problem has operated as a foundational inventory management model. Conceptually, the core challenge we are facing centers on the simultaneous lot-sizing of multiple commodities over a given planning horizon, aiming to identify a replenishment policy that minimizes long-run average operating costs. Beyond its theoretical elegance, this framework functions as a major platform for both advancing a wide range of supply chain methodologies and developing highly practical industrial solutions. Given the massive body of work along these lines, spanning rigorous methods, heuristic approaches, experimental studies, real-life applications, and software solutions, providing an exhaustive overview of relevant literature extends beyond the scope of this paper. Instead, readers are directed to extensive surveys \citep{AksoyE88, GoyalS89, MuckstadtR93, KhoujaG08, BastosMNMC17} and dedicated book chapters \citep{SilverP85, Zipkin00, MuckstadtS10, SimchiLeviCB14} for a comprehensive understanding of these developments. 

From a technical perspective, the primary algorithmic obstacle we confront in rigorous treatment of the joint replenishment problem involves efficiently synchronizing multiple economic order quantity models. Viewed in isolation, separately optimizing each commodity requires finding the precise equilibrium point between its own marginal ordering and inventory holding costs. However, once embedded within an integrated system, these commodities become intricately linked through a shared joint ordering cost, triggered every time an order is placed, irrespective of the exact subset of commodities included. Consequently, identifying an optimal policy requires a fine-grained balance, coordinating replenishment cycles to group commodities and maximize the benefit of shared setup expenses, without simultaneously inflating on-hand inventory levels. This delicate interplay between individual and shared operational costs presents fundamental analytical questions, driving continuous research into structural characterizations of rigorously analyzable lot-sizing strategies and their algorithmic implications for several decades.

\paragraph{Classical work.} To address these synchronization challenges, a series of seminal papers introduced the concept of ``power-of-$2$ policies'' \citep{Roundy85, Roundy86, JacksonMM85, MaxwellM85, MuckstadtR87}. This approach operates by rounding optimal solutions of natural convex relaxations into highly structured policies, restricting every replenishment interval to be a power of $2$ multiplied by a shared base period. When the latter period is exogenously predetermined, this mechanism guarantees a long-run average cost that falls within factor $\sqrt{9/8} \approx 1.0606$ of optimal. Furthermore, allowing the base period itself to be optimized tightens this performance bound to a factor of $\frac{1}{\sqrt{2}\ln 2} \approx 1.0201$. For decades, these findings have stood as celebrated breakthroughs in inventory management, with their elegant derivations and broad applicability extensively studied in subsequent research; see, e.g., \citet{FedergruenQZ92}, \citet{MuckstadtR93}, and \citet{TeoB01}. Consequently, a long-standing open question asked whether entirely different coordination mechanisms can outperform power-of-$2$ policies and establish even stronger approximation guarantees.

\paragraph{Recent advances.} Answering this persistent question, our recent work \citep{Segev23JRP} has fully characterized the computational landscape of continuous-time joint replenishment. Specifically, we introduced a new synchronization framework, termed $\Psi$-pairwise alignment, leading to the first polynomial-time approximation scheme, exactly matching earlier intractability results \citep{CohenHillelY18, SchulzT22}. This algorithmic method guarantees that, for any given $\eps > 0$, we can compute a replenishment policy whose long-run average cost falls within factor $1+\eps$ of optimal, requiring a running time of $O(2^{\tilde{O}(1/\eps^3)} \cdot n^{O(1)})$. Building upon this framework, our subsequent investigations \citep{Segev25improved} streamlined the analysis of $\Psi$-pairwise alignment, arriving at an improved running time of $O(2^{\tilde{O}(1/\eps)} \cdot n^{O(1)})$, while simultaneously resolving additional long-standing conjectures. Most notably, we proved that finely tuned evenly spaced policies attain an approximation guarantee of $1.0192$, strictly outperforming their power-of-$2$ counterparts, and proposed a black-box reduction that migrates variable-base performance bounds to fixed-base settings, yielding a $(\frac{1}{\sqrt{2}\ln 2} + \eps)$-approximation.

\paragraph{Resource-constrained joint replenishment.} Despite this recent progress, our alignment techniques exhibit basic limitations when transitioning from idealized models to more complex logistical realities. In particular, their structural reliance on so-called ``approximate representatives'' is very fragile when confronted with constrained generalizations of the joint replenishment problem. Most notably, it is entirely unclear whether $\Psi$-pairwise alignment retains any relevancy with the introduction of even rudimentary capacity restrictions. To bridge these critical gaps, our present work investigates the resource-constrained joint replenishment problem, whose core operational mechanism is similarly structured around the synchronization of multiple economic order quantity models. However, as formally explained in Section~\ref{subsec:model_definition}, the entire system must additionally adhere to strict upper bounds on the frequency, volume, or shared expense of replenishment events. For an elaborate discussion on the theory, practice, and rich history of these constraints, one could consult the classical work of \citet{Dobson87}, \citet{JacksonMM88}, and \citet{Roundy89}, as well as foundational investigations into capacitated lot-sizing and multi-item storage systems by \citet{Anily91}, \citet{GallegoQS96}, and \citet{GayonMRS17}. In what follows, we highlight a number of real-life environments where such capacity requirements naturally arise, especially when resource consumption is bounded as an aggregate rate over time.

In global procurement and maritime logistics, for example, commodities sourced from common overseas suppliers are frequently consolidated into standard shipping containers. Here, strict volumetric and weight limits force decision-makers to synchronize replenishment cycles to maximize container utilization without exceeding aggregate fleet capacity \citep{PorrasD08}. For in-plant logistics and Just-In-Time manufacturing, material handling equipment such as automated guided vehicles or forklifts operates under bounded bandwidth, restricting the frequency at which component batches can be transferred from storage to assembly lines \citep{MoonC06, Hoque06}. Similarly, inbound receiving operations face inspection limitations, where the total labor hours required to process incoming batches cannot exceed the available workforce capacity over a given period \citep[e.g.,][]{Dobson87, JacksonMM88}. In capacitated production-distribution networks, bottleneck machines possess a finite number of setup hours allocated per year \citep[see][]{BitranY82, Roundy89, GayonMRS17}. Finally, spatial and financial limitations play a crucial role when corporate policies restrict the aggregate storage volume or the purchasing budget dedicated to replenishment activities across a fiscal cycle \citep{Anily91, GallegoQS96}. 

To thoroughly examine the computational consequences of these capacity limitations and to properly contextualize our main contributions, we proceed by providing a complete mathematical formulation of the resource-constrained joint replenishment problem in its continuous-time setting. Subsequently, we will review cornerstone algorithmic techniques in this context, consisting of classic approximation guarantees alongside recent proof-of-concept findings. 

\subsection{Model formulation} \label{subsec:model_definition}

\paragraph{The economic order quantity model.} For a gradual exposition, let us begin by introducing the economic order quantity (EOQ) model, forming the basic building block of joint replenishment settings. Here, we are asked to identify the optimal time interval $T$ between successive orders of a single commodity, with the objective of minimizing its long-run average cost over the continuous planning horizon $[0,\infty)$. Specifically, this commodity is assumed to be characterized by a stationary demand rate of $d$, to be fully met upon occurrence, meaning that we do not allow for lost sales or back orders. In this setting, periodic policies simply fix a uniform ordering frequency across the planning horizon. Namely, orders will be placed at each of the points $0, T, 2T, 3T, \ldots$, making the time interval $T$ our sole decision variable to be optimized. To understand the intrinsic tradeoff we are facing, each of these orders incurs a fixed cost of $K > 0$ regardless of its quantity, motivating infrequent orders by itself. However, this trend is counterbalanced by a linear holding cost of $h > 0$, incurred per time unit for each inventory unit in stock, implying that we wish to avoid high inventory levels via frequent orders. As such, the fundamental question is how to determine the time interval $T$, aiming to minimize long-run average ordering and holding costs. 

Based on the preceding discussion, it is straightforward to verify that optimal policies will be placing orders only when their on-hand inventory level is completely exhausted, or put differently, zero-inventory ordering (ZIO) policies are optimal. Therefore, we can succinctly represent the objective function of interest as $C(T) = \frac{ K }{ T } + HT$, with the convention that $H = \frac{ hd }{ 2 }$. Elementary calculus arguments show that this cost function is strictly convex and that its unique minimizer is $T^* = \sqrt{ K / H }$, corresponding to the well-known square-root formula. A host of additional properties in this context can be discovered by consulting authoritative book chapters, such as those of \citet[Sec.~3]{Zipkin00}, \citet[Sec.~2]{MuckstadtS10}, or \citet[Sec.~7.1]{SimchiLeviCB14}, and we will mention particularly useful ones throughout our analysis.

% \begin{claim} \label{clm:EOQ_properties}
% The cost function $C : (0,\infty) \to \bbR_+$ satisfies the following properties:
% \begin{enumerate}
%     \item $C$ is strictly convex.

%     \item The unique minimizer of $C$ is $T^* = \sqrt{ K / H }$.

%     \item $C( \theta T^* ) = \frac{ 1 }{ 2 } \cdot (\theta + \frac{ 1 }{ \theta } ) \cdot C( T^* )$, for every $\theta > 0$.

%     \item $\min \{ C( \alpha ), C( \beta ) \} \leq \frac{ 1 }{ 2 } \cdot ( \sqrt{ \frac{ \beta }{ \alpha } } + \sqrt{ \frac{ \alpha }{ \beta } } ) \cdot C(T)$, for every $\alpha \leq \beta$ and $T \in [\alpha, \beta]$.
% \end{enumerate}
% \end{claim}

\paragraph{The joint replenishment problem.} With these foundations in place, the essence of joint replenishment can be captured by posing the next question: How should we coordinate multiple economic order quantity models, when different commodities are tied together via joint ordering costs? Specifically, we would like to synchronize the lot-sizing of $n$ distinct commodities, where each commodity $i \in [n]$ is coupled with its own EOQ model, parameterized by ordering and holding costs $K_i$ and $H_i$, respectively. As previously explained, setting a time interval of $T_i$ between successive
orders of this commodity would lead to a marginal operational cost of $C_i(T_i) = \frac{ K_i }{ T_i } + H_i T_i$. However, the complicating feature is that we concurrently face joint ordering costs, paying $K_0$ whenever an order is placed, regardless of its particular subset of commodities.

Given these parameters, any joint replenishment policy can be represented through a vector $T = (T_1, \ldots, T_n)$, where $T_i$ stands for the time interval between successive
orders of each commodity $i \in [n]$. For such policies, the first component of our objective function, $\sum_{i \in [n]} C_i( T_i )$, encapsulates the sum of marginal EOQ-based costs. Its second component, which will be designated by $J(T)$, captures long-run average joint ordering costs. Formally, this term is defined as the asymptotic density $J( T ) = K_0 \cdot \lim_{\Delta \to \infty} \frac{ N(T,\Delta) }{ \Delta }$, where $N(T,\Delta)$ stands for the number of joint orders occurring across $[0,\Delta]$ with respect to the time intervals $T_1, \ldots, T_n$. That is, letting ${\cal M}_{T_i,\Delta} = \{ 0, T_i, 2T_i, \ldots, \lfloor \frac{ \Delta }{ T_i } \rfloor \cdot T_i \}$ be the set of integer multiples of $T_i$ within $[0,\Delta]$, we have $N(T,\Delta) = | \bigcup_{i \in [n]} {\cal M}_{T_i,\Delta} |$. Putting both components together, the long-run average operational cost of any joint replenishment policy $T = (T_1, \ldots, T_n)$ can be written as $F(T) = J(T) + \sum_{i \in [n]} C_i( T_i )$.

\paragraph{Incorporating resource constraints.} A common thread passing through the preceding discussion is that different commodities are interacting ``only'' via their common ordering density. The resource-constrained joint replenishment problem introduces a further layer of complexity, where the underlying commodities are interrelated through resource consumption bounds. While serving a wide range of practical purposes, it is instructive to take a production-oriented view on this family of constraints, where assembling each commodity $i$ requires $D$ limited resources. Assuming each order of commodity $i$ requires $\alpha_{id} \geq 0$ units of resource $d \in [D]$, a time interval of $T_i$ induces a long-run average consumption rate of $\frac{\alpha_{id}}{T_i}$ for this resource. In turn, our system is constrained by ensuring the aggregate rate across all commodities does not exceed the total resource capacity, $\beta_d$. By normalizing the latter constant, such constraints can be compactly specified by asking the time intervals $T_1, \ldots, T_n$ to satisfy $\sum_{i \in [n]} \frac{ \alpha_{id} }{ T_i } \leq 1$ for every $d \in [D]$. In summary, our objective is to identify a resource-feasible joint replenishment policy $T = (T_1, \ldots, T_n)$ whose long-run average operational cost $F(T)$ is minimized. 

\subsection{Known approximability results} \label{subsec:related_work}

\paragraph{The classics: Power-of-$2$ policies.} By investigating natural convex relaxations, on which Section~\ref{subsec:conv_relax} sheds further light, the seminal work of \citet{Roundy89} established that their optimal solutions can be rounded to resource-feasible power-of-$2$ policies with bounded loss in optimality. Similar to unconstrained settings, this approach operates by expressing each time interval as a $[1, 2)$-valued base coefficient, multiplied by a power-of-$2$ component. By carefully shifting these base coefficients, Roundy generated polynomially many candidate policies, each coupled with ``ideal'' base coefficients alongside an ordering-frequency lower bound, dictated by our resource restrictions. Subsequently, to correct potential constraint violations, feasibility is restored by scaling these ideal coefficients up to their corresponding lower bound. In the presence of an arbitrary number of resources, this rounding framework guarantees a long-run average cost within factor $\frac{ 1 }{ \ln 2 } \approx 1.4427$ of the convex relaxation optimum. Furthermore, when dealing with a single resource, \citet{Roundy89} proved that this bound improves to $\sqrt{9/8} \approx 1.0606$. Building upon these foundational results, \citet{TeoB01} introduced an alternative algorithmic framework utilizing randomized rounding. Their approach recovers the above-mentioned performance guarantees by mapping convex relaxation solutions to probability distributions over power-of-$2$ policies. This randomized mechanism considerably simplified the derivation of classical findings and facilitated extensions to lot-sizing models exhibiting submodular joint ordering costs.

\paragraph{Recent progress.} Our recent work \citep{Segev25improved} further examined the approximability of resource-constrained joint replenishment by revisiting its convex relaxation, focusing on whether $\frac{ 1 }{ \ln 2 }$ is a true barrier in this domain or an intrinsic feature of existing methods. To devise a proof-of-concept improvement, we introduced a dual-policy rounding framework, which classifies commodities as being ``small'' and ``large'', depending on the magnitude of their relaxed time intervals. By developing two juxtaposed policies, via a deterministic right-shift procedure that slightly saves on the cost inflation of small commodities, and a randomized-shift procedure that leads to analogous savings for large commodities, we established an approximation factor of $1.417 < \frac{ 1 }{ \ln 2 }$.

Additionally, we investigated the computational status of resource-constrained joint replenishment instances whose number of resource constraints, $D$, is relatively small. In this setting, we proposed a discretization method that considers only $O(\frac{1}{\eps} \log \frac{n}{\eps})$ candidate time intervals for each commodity. Then, to overcome the unbounded integrality gaps inherent to standard LP-relaxations, we considered an assignment-based formulation, augmented by exhaustively guessing and fixing assignment variables for ``heavy'' resource-consuming pairs and ``expensive'' cost-incurring pairs. This mechanism allows for near-lossless randomized rounding, through which we construct a $(1+\eps)$-approximate policy in $O( n^{\tilde{O}( D^3/\eps^4 )} )$ time, implying in particular that for a constant number of resources (i.e., $D = O(1)$), our framework provides a polynomial-time approximation scheme. Notably, for the single-resource case, this outcome improves upon the classical $\sqrt{9/8}$-approximation developed by \citet{Roundy89} and \citet{TeoB01}.

\paragraph{Hardness results.} While our work is algorithmically driven, it is instructive to briefly highlight known complexity results, to better understand the theoretical landscape of resource-constrained joint replenishment. As explained as part of our concluding remarks, we still do not know whether the incorporation of such constraints makes the joint replenishment problem provably harder to approximate in comparison to its unconstrained setting. In the latter context, \citet{SchulzT11} were the first to rigorously examine the computational feasibility of obtaining optimal periodic policies in continuous time. Specifically, they established that in the fixed-base setting, an exact polynomial-time algorithm for the joint replenishment problem would imply a similar outcome for integer factorization, thereby establishing novel connections to core challenges in number theory. Subsequently, \citet{CohenHillelY18} proved that the fixed-base setting is strongly NP-hard. This result was strengthened by \citet{SchulzT22}, who demonstrated that NP-hardness persists even for two-commodity instances. Extending their analysis to the variable-base setting, they demonstrated its polynomial-time reducibility to integer factorization. This reduction was subsequently lifted to a strong NP-hardness proof by \citet{TuisovY20}.

\subsection{Main contributions}

The primary contribution of this paper resides in devising a systematic rounding framework that significantly surpasses the best-known approximation guarantees for resource-constrained joint replenishment. While previous approaches solely relied on bounding the cost inflation inherent to power-of-$2$ policies, the present work examines a much broader class of replenishment policies. In particular, the technical domains we study and their analysis enhance traditional ideas by encompassing fractional expansion factors, randomized shifting, and staggered interleaved grids. As outlined below, these developments culminate in the following performance guarantee for general problem instances, with an arbitrary number of capacity constraints.

\begin{theorem} \label{thm:main_result}
The resource-constrained joint replenishment problem can be approximated in polynomial time within factor $\frac{5}{6\ln 2} < 1.2023$ of optimal.
\end{theorem}

\paragraph{Static power-of-$m^{1/k}$ policies.} In Section~\ref{sec:unshifted}, we initiate our algorithmic progression by rounding optimal solutions of a well-known convex relaxation via a static geometric grid, defined by a possibly fractional expansion factor of $m^{1/k}$. By identifying the optimal grid parameters ($m=2$, $k=2$), we establish a baseline approximation guarantee of $\sqrt{2} < 1.4143$ for the resulting policy, $T^{2,2}$. To push beyond this limit, we exploit a structural asymmetry within the rounding operator in play, as its respective multipliers governing joint ordering density and inventory holding costs are unbalanced. Specifically, our optimal single-shot policy $T^{2,2}$ incurs a highly efficient density multiplier alongside a heavy holding cost multiplier, whereas the alternative configuration $T^{2,3}$ exhibits the inverse cost structure. By carefully pairing these policies and selecting the cheaper of two static configurations, we obtain an improved approximation ratio of $1.3776$.

\paragraph{Randomized shifting.} In Section~\ref{sec:randomized_power_of_mk}, to further mitigate the rounding errors incurred by static grids, we inject a uniformly distributed shift $\Theta \sim U[0,1]$ into the logarithmic domain of the expansion factor. In essence, scaling grid elements by $m^{\Theta/k}$ smooths the expected individual ordering and holding cost multipliers across the continuum of $\Theta$ values. We show that an optimal single-shot randomized policy, $T^{3,2,\Theta}$, leverages such smoothing to attain an expected approximation guarantee of $\frac{2(\sqrt{3}-1)}{\ln 3} < 1.3327$. Moreover, exploiting the asymmetry between two randomized policies ($T^{2,1,\Theta}$ and $T^{2,2,\Theta}$) further drops the expected bound to $1.2585$. Finally, moving beyond oblivious randomization, we transition to a proactive instance-dependent selection of $\Theta$. By discretizing our shift parameter and formulating a factor-revealing linear program to evaluate worst-case adversarial cost distributions, we arrive at a strictly better $1.2512$-approximation.

\paragraph{Interleaved policies.} In Section~\ref{sec:Interleaved}, we tackle the expected holding cost inflation inherent to any single geometric grid. For this purpose, we introduce a secondary grid, scaled by a fractional offset of $\alpha$, which is superimposed onto the primary grid to form a denser interleaved structure. This approach subdivides the primary grid's rounding intervals, restricting the maximum possible rounding error for any given commodity. While utilizing a secondary grid inflates joint ordering costs, optimizing the offset $\alpha$ to stagger ordering points ensures that integer multiples of these two grids frequently synchronize. Instantiating this interleaved framework with $m=2$, $k=1$, and $\alpha=3/2$ strikes an effective balance between bounding holding costs and suppressing ordering density blowup. This configuration establishes our ultimate expected approximation ratio of $\frac{5}{6\ln 2} < 1.2023$, which is proven to be best-possible for the entire class of interleaved power-of-$m^{1/k}$ policies.
\section{\texorpdfstring{Static Power-of-$m^{1/k}$ Policies}{}} \label{sec:unshifted}

In this section, we introduce and study the class of static power-of-$m^{1/k}$ policies. This concept is based on deterministically rounding optimal solutions to a well-known convex relaxation of the resource-constrained joint replenishment problem, utilizing a geometric grid with an expansion factor of $m^{1/k}$. Our first finding along these lines shows that by optimally tuning grid parameters, a single-shot power-of-$m^{1/k}$ policy suffices to achieve a worst-case performance guarantee of $\sqrt{2} < 1.4143$. Second, we exploit the structural asymmetry between the ordering density and holding cost multipliers of our rounding procedure: By selecting the cheaper of two carefully matched policies, we obtain an improved approximation ratio of $1.3776$.

\subsection{The convex relaxation} \label{subsec:conv_relax}

Our mechanism for lower-bounding optimal long-run average operational costs is similar to that of \citet{Roundy89} and \citet{TeoB01}, who considered the following convex relaxation:
\begin{alignat}{3} 
& \text{min} \quad && \frac{ K_0 }{ T_{\min} } + \sum_{i \in [n]} \left( \frac{ K_i }{ T_i } + H_i T_i \right) \tag{P} \label{eqn:conv_relax_RCJRP} \\
&\text{s.t.} \quad && T_i \geq T_{\min} \geq 0 \quad && \forall \, i \in [n] \label{eqn:conv_relax_RCJRP_1} \\
& \quad && \sum_{i \in [n]} \frac{ \alpha_{id} }{ T_i } \leq 1 && \forall \, d \in [D] \label{eqn:conv_relax_RCJRP_2}
\end{alignat}
Reading in reverse order, constraint~\eqref{eqn:conv_relax_RCJRP_2} ensures that the time intervals $T_1, \ldots, T_n$ are resource-feasible. Constraint~\eqref{eqn:conv_relax_RCJRP_1} involves the auxiliary decision variable $T_{\min}$, which represents $\min_{ i \in [n] } T_i$ in an optimal solution, since the objective function of~\eqref{eqn:conv_relax_RCJRP} is strictly decreasing in $T_{\min}$. Finally, to derive a closed-form bound on the operational cost $F(T) = J(T) + \sum_{i \in [n]} C_i( T_i )$ of any replenishment policy, this formulation replaces the asymptotic ordering density $J(T)$ with $\frac{ K_0 }{ T_{\min} }$. This term coincides with the individual density of the most frequently ordered commodity, which provides a straightforward lower bound on $J(T)$.

Summarizing common knowledge, we point out that formulation~\eqref{eqn:conv_relax_RCJRP} is a convex program that can be solved to optimality in polynomial time via interior-point methods. Specifically, by transforming our decision variables into frequency-based ones, where $f_i = \frac{ 1 }{ T_i }$, this problem can be cast as a second-order cone program or as a convex geometric program. For a comprehensive overview of such problems and their algorithmic foundations, we refer the reader to \citet[Chap.~4]{Nesterov94} for interior-point theory and to \citet[Sec.~4.4--4.6]{BoydVandenberghe04} for geometric and cone programming. Moving forward, $T^* = (T_{\min}^*, T_1^*, \ldots, T_n^*)$ will designate an optimal solution to~\eqref{eqn:conv_relax_RCJRP}.

\subsection{\texorpdfstring{Power-of-$m^{1/k}$ policies: Preliminaries}{}} \label{subsec:det_mk_prelim}

To describe the mechanics of power-of-$m^{1/k}$ policies, we begin by creating a deterministic geometric grid and elaborating on its induced rounding operator and density coefficient. For a gradual exposition of these concepts, the objects we will be interacting with are formally defined as follows.

\paragraph{The geometric grid.} For a pair of integers $m \geq 2$ and $k \geq 1$ whose values will be optimized later on, let $m^{1/k}$ be the expansion factor of our grid. We restrict our attention to this discrete parametrization, noting that the set $\{ m^{1/k} : m \in \bbZ_{\geq 2}, k \in \bbZ_{\geq 1} \}$ is dense in the interval $(1,\infty)$, thereby covering arbitrarily valued expansion factors via potentially irrational terms. Given $m$, $k$, and $T_{\min}^*$, we consider the geometric grid ${\cal G}_{m,k}$, where $T_{\min}^*$ is multiplied by all integer powers of the expansion factor $m^{1/k}$, namely, ${\cal G}_{m,k} = \{ m^{p/k} \cdot T_{\min}^* : p \in \bbZ \}$. 

\paragraph{The operator $\myround_{m,k}$.} With respect to this grid, the operator $\myround_{m,k} : \bbR_{++} \to \bbR$ rounds its argument up to the nearest strictly greater point in ${\cal G}_{m,k}$, meaning that $\myround_{m,k}( t ) = \min \{ g \in {\cal G}_{m,k} : g > t \}$. We emphasize that this operator is purposely defined with a strict inequality, implying in particular that grid points are not rounded to themselves but rather to their successors. As explained in the sequel, the latter property is an important design feature for our current class of policies. 

\paragraph{The density coefficient ${\cal D}_{m,k}$.} To introduce our next component, which serves an analytical role in bounding joint ordering costs, for any element $g \in {\cal G}_{m,k}$ and $\Delta \geq 0$, let ${\cal M}_{g,\Delta} = \{ 0, g, 2g, \ldots, \lfloor \frac{ \Delta }{ g } \rfloor \cdot g \}$ be the collection of integer multiples of $g$ within the interval $[0,\Delta]$. Then, the density coefficient ${\cal D}_{m,k}$ of ${\cal G}_{m,k}$ is defined as the asymptotic density of the union of these integer multiples, over elements of value at least $\myround_{m,k}( T_{\min}^* )$, i.e., 
\begin{equation} \label{eqn:def_density_det}
{\cal D}_{m,k} ~~=~~ \lim_{\Delta \to \infty} \frac{ | \bigcup_{g \in {\cal G}_{m,k} \cap [\myround_{m,k}( T_{\min}^* ), \infty) } {\cal M}_{g,\Delta} | }{ \Delta } \ . 
\end{equation}
Even though this union is taken over infinitely many sets, it is identical to the finite union of its first $k$ sets, ${\cal M}_{ \myround_{m,k}( T_{\min}^* ), \Delta}, {\cal M}_{m^{1/k} \cdot \myround_{m,k}( T_{\min}^* ),\Delta}, \ldots, {\cal M}_{m^{(k-1)/k} \cdot \myround_{m,k}( T_{\min}^* ),\Delta}$. Indeed, since $m$ is an integer, ${\cal M}_{mg,\Delta} \subseteq {\cal M}_{g,\Delta}$ for every $g \in {\cal G}_{m,k}$, effectively making all sets beyond the first $k$ redundant.

\subsection{Rounding procedure} \label{subsec:det_rounding}

\paragraph{Creating $T^{m,k}$.} Given an optimal solution $T^* = (T_{\min}^*, T_1^*, \ldots, T_n^*)$ to the convex relaxation~\eqref{eqn:conv_relax_RCJRP}, we round this vector via the $\myround_{m,k}$-operator into a power-of-$m^{1/k}$ policy, $T^{m,k}$, setting $T^{m,k}_i = \myround_{m,k}( T_i^* )$ for every commodity $i \in [n]$. As such, since our resource constraints $\{ \sum_{i \in [n]} \frac{ \alpha_{id} }{ T_i } \leq 1 \}_{d \in [D]}$ are upward-closed, and since $\myround_{m,k}$ inflates its argument (i.e., $\myround_{m,k}(t) > t$), the replenishment policy $T^{m,k}$ is clearly resource-feasible.

\paragraph{Parametric performance guarantee.} For the purpose of evaluating the operational cost $F( T^{m,k} )$ of this policy, we observe that the time interval $T^{m,k}_i$ of every commodity must be a grid point in ${\cal G}_{m,k} \cap [\myround_{m,k}( T_{\min}^* ), \infty)$. Therefore, $T^{m,k}$ places joint orders only at integer multiples of these points, precisely corresponding to the union in definition~\eqref{eqn:def_density_det} of the density coefficient ${\cal D}_{m,k}$, implying that $J( T^{m,k} ) \leq K_0 \cdot {\cal D}_{m,k}$. To resolve overlapping integer multiples of relevant grid elements, we derive an inclusion-exclusion-type expression for this coefficient in Claim~\ref{clm:exact_density}, whose proof is presented in Section~\ref{subsec:proof_clm_exact_density}. Here, for any finite set ${\cal N} \subseteq \bbR_{++}$, the smallest positive value that is an integer multiple of every element in ${\cal N}$ will be denoted by $\mylcm({\cal N})$; when these elements are incommensurable, we set $\mylcm({\cal N}) = \infty$.

\begin{claim} \label{clm:exact_density}
${\cal D}_{m,k} = \frac{ 1 }{ \myround_{m,k}( T_{\min}^* ) } \cdot \sum_{\emptyset \neq {\cal N} \subseteq {\cal S}_{m,k}} \frac{ (-1)^{ |{\cal N}| + 1} }{ \mylcm({\cal N}) }$, where ${\cal S}_{m,k} = \{ m^{(\kappa-1)/k} \}_{\kappa \in [k]}$.
\end{claim}

This observation allows us to relate the operational cost of $T^{m,k}$ to the convex relaxation optimum in terms of the grid parameters $m$ and $k$, by noting that
\begin{align}
F(T^{m,k}) ~~\leq~~ & K_0 \cdot {\cal D}_{m,k} + \sum_{i \in [n]} \left( \frac{ K_i }{ T^{m,k}_i } + H_i T^{m,k}_i \right) \nonumber \\
~~\leq~~ & \left( \sum_{\emptyset \neq{\cal N} \subseteq {\cal S}_{m,k}} \frac{ (-1)^{ |{\cal N}| + 1} }{ \mylcm({\cal N}) } \right) \cdot \frac{ K_0 }{ \myround_{m,k}( T_{\min}^* ) } + \sum_{i \in [n]} \left( \frac{ K_i }{ \myround_{m,k}( T_i^* ) } + H_i \myround_{m,k}( T_i^* ) \right) \nonumber \\
~~\leq~~ & \frac{ 1 }{ m^{1/k} } \cdot \left( \sum_{\emptyset \neq{\cal N} \subseteq {\cal S}_{m,k}} \frac{ (-1)^{ |{\cal N}| + 1} }{ \mylcm({\cal N}) } \right) \cdot \frac{ K_0 }{ T_{\min}^* } + \sum_{i \in [n]} \frac{ K_i }{ T_i^* } + m^{1/k} \cdot \sum_{i \in [n]} H_i T_i^* \ . \label{eqn:bound_cost_det}
\end{align}
The last inequality holds since $T_i^* < \myround_{m,k}( T_i^* ) \leq m^{1/k} \cdot T_i^*$ for every commodity $i \in [n]$. In addition, $\myround_{m,k}( T_{\min}^* ) = m^{1/k} \cdot T_{\min}^*$, which would not have been the case had $\myround_{m,k}$ been defined with a weak inequality. 

\subsection{Policy instantiation} \label{subsec:policies_det_inst}

\paragraph{Optimal single-shot policy.} To extract the best-possible approximation guarantee out of this mechanism, we proceed by optimizing the grid parameters $m$ and $k$, depending on whether $m^{1/k}$ is rational or not. First, by standard divisibility arguments, $m^{1/k}$ is rational if and only if $m$ is a perfect $k$-th power of an integer. In this case, the expansion factor $m^{1/k}$ is integer-valued, and we may therefore assume $k=1$ without loss of generality. In turn, letting $T^{m,1}$ stand for the resulting policy, our upper bound~\eqref{eqn:bound_cost_det} specializes to 
\[ F(T^{m,1}) ~~\leq~~ \frac{ 1 }{ m } \cdot \frac{ K_0 }{ T_{\min}^* } + \sum_{i \in [n]} \frac{ K_i }{ T_i^* } + m \cdot \sum_{i \in [n]} H_i T_i^* \ , \]
implying that such a policy does not improve on the trivial $2$-approximation, since $m \geq 2$. 

Now, when $m^{1/k}$ is irrational, we observe that the inclusion-exclusion term in~\eqref{eqn:bound_cost_det} collapses to the union bound, as ${\cal S}_{m,k} = \{ m^{(\kappa-1)/k} \}_{\kappa \in [k]}$ consists of incommensurable elements. Consequently, $\sum_{\emptyset \neq{\cal N} \subseteq {\cal S}_{m,k}} \frac{ (-1)^{ |{\cal N}| + 1} }{ \mylcm({\cal N}) } = \sum_{\kappa \in [k]} \frac{ 1 }{ m^{(\kappa-1)/k} } = \frac{ 1 - 1/m }{ 1 - m^{-1/k} }$, leading to a simplified upper bound of
\begin{equation} \label{eqn:bound_det_irrational}
F(T^{m,k}) ~~\leq~~ \underbrace{ \left( 1 - \frac{ 1 }{ m } \right) \cdot \frac{ 1 }{ m^{1/k}-1 } }_{ V_J(m,k) } \cdot \frac{ K_0 }{ T_{\min}^* } + \sum_{i \in [n]} \frac{ K_i }{ T_i^* } + \underbrace{ \vphantom{\left( 1 - \frac{ 1 }{ m } \right) \cdot \frac{ 1 }{ m^{1/k}-1 }} m^{1/k} }_{ V_H(m,k) } \cdot \sum_{i \in [n]} H_i T_i^* \ .
\end{equation}
To derive the best approximation of this form, we wish to determine $m \in \bbZ_{\geq 2}$ and $k \in \bbZ_{\geq 1}$ that minimize the larger coefficient out of $V_J(m,k)$ and $V_H(m,k)$. The next claim, whose proof is provided in Appendix~\ref{app:proof_clm_optimal_base_unshift}, shows that the optimal choice is $m^* = k^* = 2$. 

\begin{claim} \label{clm:optimal_base_unshift}
The integers $m \in \bbZ_{\geq 2}$ and $k \in \bbZ_{\geq 1}$ that minimize $\max \{ V_J(m,k), V_H(m,k) \}$ are $m^* = 2$ and $k^* = 2$, yielding an optimum value of $\sqrt{2}$.
\end{claim}

We have just shown that the policy obtained, $T^{2,2}$, is associated with a long-run average operational cost within factor $\sqrt{2} < 1.4143$ of optimal. This result by itself is already an improvement over the $1.417$-approximation attained by our previous work in this context \citep{Segev25improved}. 

\begin{lemma} \label{lem:approx_det_2_2}
$F( T^{2,2} ) \leq \sqrt{2} \cdot \opt\eqref{eqn:conv_relax_RCJRP} $.
\end{lemma}

\paragraph{Best-of-both-worlds policy.} Our next enhancement is derived by noting that, while $m^* = k^* = 2$ is an optimal choice in the sense of minimizing $\max \{ V_J(m,k), V_H(m,k) \}$, these two terms are not perfectly balanced, due to the integrality restriction on $m$ and $k$. Specifically, in light of the irrational-case bound~\eqref{eqn:bound_det_irrational}, 
\begin{equation} \label{eqn:def_FT_22}
F(T^{2,2}) ~~\leq~~ \frac{ \sqrt{2} + 1}{ 2 } \cdot \frac{ K_0 }{ T_{\min}^* } + \sum_{i \in [n]} \frac{ K_i }{ T_i^* } + \sqrt{2} \cdot \sum_{i \in [n]} H_i T_i^* \ , 
\end{equation}
where the ordering density component has a multiplier of $\frac{ \sqrt{2} + 1}{ 2 } \approx 1.2071$, much better than the $\sqrt{2}$-multiplier of the holding cost component. To exploit this structure, we consider the policy $T^{2,3}$, which exhibits the inverse asymmetry. Here,
\begin{equation} \label{eqn:def_FT_23}
F(T^{2,3}) ~~\leq~~ \frac{ 1}{ 2(2^{1/3}-1) } \cdot \frac{ K_0 }{ T_{\min}^* } + \sum_{i \in [n]} \frac{ K_i }{ T_i^* } + 2^{1/3} \cdot \sum_{i \in [n]} H_i T_i^* \ , 
\end{equation}
meaning that we incur an ordering density multiplier of $\frac{ 1}{ 2(2^{1/3}-1) } \approx 1.9236$, in contrast to a holding cost multiplier of only $2^{1/3}$, which is strictly better than $\sqrt{2}$. Consequently, by picking the cheaper of these two policies, we improve on the worst-case performance guarantee of $T^{2,2}$ by itself, as shown below.

\begin{lemma} \label{lem:approx_better_2det}
$\min \{ F( T^{2,2} ), F( T^{2,3} ) \} \leq 1.3776 \cdot \opt\eqref{eqn:conv_relax_RCJRP} $.
\end{lemma}
\begin{proof}
By synthesizing the upper bounds~\eqref{eqn:def_FT_22} and~\eqref{eqn:def_FT_23} on the operational costs of $T^{2,2}$ and $T^{2,3}$, respectively, we conclude that
\begin{align}
\min & \{ F( T^{2,2} ), F( T^{2,3} ) \} \nonumber \\
& \leq~~ \min_{\lambda \in [0,1]} \left\{ \lambda \cdot \left( \frac{ \sqrt{2} + 1}{ 2 } \cdot \frac{ K_0 }{ T_{\min}^* } + \sum_{i \in [n]} \frac{ K_i }{ T_i^* } + \sqrt{2} \cdot \sum_{i \in [n]} H_i T_i^* \right) \right. \nonumber \\
& \qquad \qquad \qquad \left. \mbox{} + (1-\lambda) \cdot \left( \frac{ 1}{ 2(2^{1/3}-1) } \cdot \frac{ K_0 }{ T_{\min}^* } + \sum_{i \in [n]} \frac{ K_i }{ T_i^* } + 2^{1/3} \cdot \sum_{i \in [n]} H_i T_i^* \right) \right\} \nonumber \\
& \leq~~ 0.76218 \cdot \left( \frac{ \sqrt{2} + 1}{ 2 } \cdot \frac{ K_0 }{ T_{\min}^* } + \sum_{i \in [n]} \frac{ K_i }{ T_i^* } + \sqrt{2} \cdot \sum_{i \in [n]} H_i T_i^* \right) \nonumber \\
& \qquad \mbox{} + (1-0.76218) \cdot \left( \frac{ 1}{ 2(2^{1/3}-1) } \cdot \frac{ K_0 }{ T_{\min}^* } + \sum_{i \in [n]} \frac{ K_i }{ T_i^* } + 2^{1/3} \cdot \sum_{i \in [n]} H_i T_i^* \right) \label{eqn:lem_approx_better_2det_1} \\
& \leq~~ 1.3776 \cdot \left( \frac{ K_0 }{ T_{\min}^* } + \sum_{i \in [n]} \left( \frac{ K_i }{ T_i^* } + H_i T_i^* \right) \right) \nonumber \\
& =~~ 1.3776 \cdot \opt\eqref{eqn:conv_relax_RCJRP} \ . \nonumber
\end{align}
We remark that the constant $0.76218$ appearing in inequality~\eqref{eqn:lem_approx_better_2det_1} is a truncation of the one that balances between the density and holding multipliers, i.e., $\frac{ \sqrt{2} + 1}{ 2 } \cdot \lambda + \frac{ 1}{ 2(2^{1/3}-1) } \cdot (1-\lambda) = \sqrt{2} \cdot \lambda + 2^{1/3} \cdot (1-\lambda)$.
\end{proof}

\subsection{Proof of Claim~\ref{clm:exact_density}} \label{subsec:proof_clm_exact_density}

In order to compute ${\cal D}_{m,k}$, according to definition~\eqref{eqn:def_density_det}, we have
\begin{align}
{\cal D}_{m,k} ~~=~~ & \lim_{\Delta \to \infty} \frac{ | \bigcup_{g \in {\cal G}_{m,k} \cap [\myround_{m,k}( T_{\min}^* ), \infty) } {\cal M}_{g,\Delta} | }{ \Delta } \nonumber \\
=~~ & \lim_{\Delta \to \infty} \frac{ | \bigcup_{\kappa \in [k] } {\cal M}_{m^{(\kappa-1)/k} \cdot \myround_{m,k}( T_{\min}^* ),\Delta} | }{ \Delta } \label{eqn:lem_bound_unshifted_1} \\
=~~ & \sum_{\emptyset \neq{\cal N} \subseteq [k] } (-1)^{ |{\cal N}|+1 } \cdot \lim_{\Delta \to \infty} \frac{ | \bigcap_{\kappa \in {\cal N}} {\cal M}_{m^{(\kappa-1)/k} \cdot \myround_{m,k}( T_{\min}^* ),\Delta} | }{ \Delta } \label{eqn:lem_bound_unshifted_2} \\
=~~ & \frac{ 1 }{ \myround_{m,k}( T_{\min}^* ) } \cdot \sum_{\emptyset \neq{\cal N} \subseteq {\cal S}_{m,k}} \frac{ (-1)^{ |{\cal N}| + 1} }{ \mylcm({\cal N}) } \ . \label{eqn:lem_bound_unshifted_3}
\end{align}
Here, equality~\eqref{eqn:lem_bound_unshifted_1} holds since the union $\bigcup_{g \in {\cal G}_{m,k} \cap [\myround_{m,k}( T_{\min}^* ), \infty) } {\cal M}_{g,\Delta}$ is identical to the finite union of its first $k$ sets, $\{ {\cal M}_{ m^{(\kappa-1)/k} \cdot\myround_{m,k}( T_{\min}^* ), \Delta } \}_{\kappa \in [k]}$, as explained in Section~\ref{subsec:det_mk_prelim}. Equality~\eqref{eqn:lem_bound_unshifted_2} follows from the inclusion-exclusion principle. Finally, recalling that ${\cal S}_{m,k} = \{ m^{(\kappa-1)/k} \}_{\kappa \in [k]}$, equality~\eqref{eqn:lem_bound_unshifted_3} is obtained by observing that
\[ \left| \bigcap_{\kappa \in {\cal N}} {\cal M}_{m^{(\kappa-1)/k} \cdot \myround_{m,k}( T_{\min}^* ),\Delta} \right| ~~=~~ \left\lfloor \frac{ \Delta }{ \myround_{m,k}( T_{\min}^* ) \cdot \mylcm( \{ m^{(\kappa-1)/k} \}_{\kappa \in {\cal N}} ) } \right\rfloor + 1 \ . \]
\section{\texorpdfstring{Shifted Power-of-$m^{1/k}$ Policies}{}} \label{sec:randomized_power_of_mk}

In what follows, we expand our algorithmic framework by introducing the class of shifted power-of-$m^{1/k}$ policies. Building upon the geometric rounding approach of Section~\ref{sec:unshifted}, this concept injects a randomized shift into the logarithmic grid domain to smooth out rounding errors. Our first finding demonstrates that, by optimally tuning grid parameters, a single-shot policy is capable of attaining an expected performance guarantee of $\frac{2(\sqrt{3}-1)}{\ln 3} < 1.3327$. Second, we once again exploit the structural asymmetry between the ordering density and holding cost multipliers of our rounding procedure. By picking the cheaper of two carefully constructed randomized policies, we obtain an improved expected approximation ratio of $1.2585$. Finally, we move beyond oblivious randomization to a proactive instance-dependent shift selection. Utilizing a factor-revealing linear program, we prove that by choosing an optimal shift for each specific instance, a tighter approximation guarantee of $1.2512$ can be established.

\subsection{Introducing randomization} \label{subsec:rand_mk_prelim}

To take the performance guarantees of Section~\ref{sec:unshifted} one step further, we augment power-of-$m^{1/k}$ policies with randomized shifting. Specifically, by revisiting earlier definitions and notation, our geometric grid, rounding operator, and density coefficient will incorporate randomization as follows.

\paragraph{The geometric grid.} We remind the reader that given $m \in \bbZ_{\geq 2}$, $k \in \bbZ_{\geq 1}$, and $T_{\min}^*$, the static geometric grid ${\cal G}_{m,k} = \{ m^{p/k} \cdot T_{\min}^* : p \in \bbZ \}$ was defined through multiplying $T_{\min}^*$ by all integer powers of the expansion factor $m^{1/k}$. Moving forward, rather than considering ${\cal G}_{m,k}$, we scale its elements by a factor of $m^{\Theta/k}$, where $\Theta \sim U[0,1]$ represents a uniformly distributed shift in the logarithmic domain, thereby creating the random grid ${\cal G}_{m,k, \Theta} = \{ m^{(p + \Theta)/k} \cdot T_{\min}^* : p \in \bbZ \}$.

\paragraph{The operator $\myround_{m,k, \Theta}$ and density coefficient ${\cal D}_{m,k, \Theta}$.} The first consequence of this extension is that, when any $t > 0$ is rounded up to the strictly greater nearest point in ${\cal G}_{m,k, \Theta}$, the latter point is clearly random. As such, our rounding operator becomes the random function $\myround_{m,k, \Theta}( t ) = \min \{ g \in {\cal G}_{m,k, \Theta} : g > t \}$. Similarly, once we take all integer multiples of grid elements valued at least $\myround_{m,k, \Theta}( T_{\min}^* )$, their asymptotic density is random as well, with the density coefficient of ${\cal G}_{m,k, \Theta}$ prescribed by
\begin{equation} \label{eqn:def_density_rand}
{\cal D}_{m,k, \Theta} ~~=~~ \lim_{\Delta \to \infty} \frac{ | \bigcup_{g \in {\cal G}_{m,k, \Theta} \cap [\myround_{m,k, \Theta}( T_{\min}^* ), \infty) } {\cal M}_{g,\Delta} | }{ \Delta } \ . 
\end{equation} 

\paragraph{Creating $T^{m,k, \Theta}$.} Given an optimal solution $T^* = (T_{\min}^*, T_1^*, \ldots, T_n^*)$ to the convex relaxation~\eqref{eqn:conv_relax_RCJRP}, we define the randomized replenishment policy $T^{m,k, \Theta}$ via the $\myround_{m,k, \Theta}$-operator, specifically meaning that $T^{m,k, \Theta}_i = \myround_{m,k, \Theta}( T_i^* )$ for every commodity $i \in [n]$. As before, since $\myround_{m,k, \Theta}$ inflates its argument, the policy $T^{m,k, \Theta}$ is resource-feasible with probability $1$.

\paragraph{Parametric performance guarantee.} In order to upper-bound the expected operational cost of this policy, we observe that $T^{m,k, \Theta}$ may place joint orders only at integer multiples of the points in ${\cal G}_{m,k, \Theta} \cap [\myround_{m,k, \Theta}( T_{\min}^* ), \infty)$. This set is clearly contained within the union appearing in definition~\eqref{eqn:def_density_rand}, implying that $J( T^{m,k, \Theta} ) \leq K_0 \cdot {\cal D}_{m,k, \Theta}$ almost surely. As stated in Claim~\ref{clm:prop_shifted}, we show that the expected density coefficient $\exsubpar{ \Theta }{ {\cal D}_{m,k, \Theta} }$ admits a closed-form expression, and so do $\exsubpar{ \Theta }{ \myround_{m,k, \Theta}(t) }$ and $\exsubpar{ \Theta }{ \frac{ 1 }{ \myround_{m,k, \Theta}(t) } }$, for any $t > 0$. The derivation of these expressions is presented in Section~\ref{subsec:proof_clm_prop_shifted}.

\begin{claim} \label{clm:prop_shifted}
\begin{enumerate}
    \item $\exsubpar{ \Theta }{ \myround_{m,k, \Theta}(t) } = \frac{ k(m^{1/k}-1)}{ \ln m} \cdot t$, for every $t > 0$. \label{item:prop_shifted_1} 

    \item $\exsubpar{ \Theta }{ \frac{ 1 }{ \myround_{m,k, \Theta}(t) } } = \frac{ k(m^{1/k}-1)}{ m^{1/k}\ln m} \cdot \frac{ 1 }{ t }$, for every $t > 0$. \label{item:prop_shifted_2} 

    \item $\exsubpar{ \Theta }{ {\cal D}_{m,k, \Theta} } = \frac{ k(m^{1/k}-1)}{ m^{1/k}\ln m} \cdot (\sum_{\emptyset \neq{\cal N} \subseteq {\cal S}_{m,k}} \frac{ (-1)^{ |{\cal N}| + 1} }{ \mylcm({\cal N}) } ) \cdot \frac{ 1 }{ T_{\min}^* }$. \label{item:prop_shifted_3} 
\end{enumerate}
\end{claim}

As an immediate byproduct, the expected operational cost of $T^{m,k, \Theta}$ can be related to the convex relaxation optimum in terms of the grid parameters $m$ and $k$, by noting that
\begin{align}
\exsub{ \Theta }{ F(T^{m,k, \Theta}) } ~~\leq~~ & K_0 \cdot \exsub{ \Theta }{ {\cal D}_{m,k, \Theta}} + \sum_{i \in [n]} \left( K_i \cdot \exsub{ \Theta }{ \frac{ 1 }{ T^{m,k, \Theta}_i } } + H_i \cdot \exsub{ \Theta }{ T^{m,k, \Theta}_i } \right) \nonumber \\
=~~ & K_0 \cdot \exsub{ \Theta }{ {\cal D}_{m,k, \Theta}} + \sum_{i \in [n]} \left( K_i \cdot \exsub{ \Theta }{ \frac{ 1 }{ \myround_{m,k, \Theta}( T_i^* ) } } + H_i \cdot \exsub{ \Theta }{ \myround_{m,k, \Theta}( T_i^* ) } \right) \nonumber \\
\leq~~ & \frac{ k(m^{1/k}-1)}{ m^{1/k}\ln m} \cdot \left( \sum_{\emptyset \neq{\cal N} \subseteq {\cal S}_{m,k}} \frac{ (-1)^{ |{\cal N}| + 1} }{ \mylcm({\cal N}) } \right) \cdot \frac{ K_0 }{ T_{\min}^* } \nonumber \\
& \mbox{} + \frac{ k(m^{1/k}-1)}{ m^{1/k}\ln m} \cdot \sum_{i \in [n]} \frac{ K_i }{ T_i^* } + \frac{ k(m^{1/k}-1)}{ \ln m} \cdot \sum_{i \in [n]} H_i T_i^* \ . \label{eqn:bound_shifted}
\end{align}

\subsection{Policy instantiation} \label{subsec:instance_rand}

\paragraph{Optimal single-shot policy.} Building on this bound, we observe that with respect to the ordering cost multiplier, $\frac{ k(m^{1/k}-1)}{ m^{1/k}\ln m} < 1$ for any $m \geq 2$ and $k \geq 1$. Therefore, we focus on identifying the grid parameters $m \in \bbZ_{\geq 2}$ and $k \in \bbZ_{\geq 1}$ that minimize the larger of the density multiplier $\frac{ k(m^{1/k}-1)}{ m^{1/k}\ln m} \cdot ( \sum_{\emptyset \neq{\cal N} \subseteq {\cal S}_{m,k}} \frac{ (-1)^{ |{\cal N}| + 1} }{ \mylcm({\cal N}) } )$ and the holding cost multiplier $\frac{ k(m^{1/k}-1)}{ \ln m }$. Similarly to the case analysis in Section~\ref{subsec:policies_det_inst}, when $m^{1/k}$ is rational, we may assume without loss of generality that $k=1$. As such, $\sum_{\emptyset \neq{\cal N} \subseteq {\cal S}_{m,1}} \frac{ (-1)^{ |{\cal N}| + 1} }{ \mylcm({\cal N}) } = 1$, and in relation to our resulting policy $T^{m,1,\Theta}$, the upper bound~\eqref{eqn:bound_shifted} becomes
\begin{equation} \label{eqn:bound_rand_rational}
\exsub{ \Theta }{ F(T^{m,1,\Theta}) } ~~\leq~~ \frac{ m-1}{ m\ln m} \cdot \frac{ K_0 }{ T_{\min}^* } + \frac{ m-1}{ m\ln m} \cdot \sum_{i \in [n]} \frac{ K_i }{ T_i^* } + \frac{ m-1}{ \ln m} \cdot \sum_{i \in [n]} H_i T_i^* \ . 
\end{equation}
An optimal policy of this form picks $m=2$, for which $\exsubpar{ \Theta }{ F(T^{2,1,\Theta}) } \leq \frac{ 1 }{ \ln 2 } \cdot \opt\eqref{eqn:conv_relax_RCJRP}$. 

On the other hand, when $m^{1/k}$ is irrational, we know that $\sum_{\emptyset \neq{\cal N} \subseteq {\cal S}_{m,k}} \frac{ (-1)^{ |{\cal N}| + 1} }{ \mylcm({\cal N}) } = \frac{ 1 - 1/m }{ 1 - m^{-1/k} }$, leading to a simplified upper bound of
\[ \exsub{ \Theta }{ F(T^{m,k,\Theta}) } ~~\leq~~ \underbrace{ \left( 1 - \frac{ 1 }{ m } \right) \cdot \frac{ k }{ \ln m } }_{ V_J(m,k) } \cdot \frac{ K_0 }{ T_{\min}^* } + \underbrace{ \vphantom{\left( 1 - \frac{ 1 }{ m } \right)} \frac{ k(m^{1/k}-1)}{ m^{1/k}\ln m} }_{ \leq 1 } \cdot \sum_{i \in [n]} \frac{ K_i }{ T_i^* } + \underbrace{ \vphantom{\left( 1 - \frac{ 1 }{ m } \right)} \frac{ k(m^{1/k}-1)}{ \ln m} }_{ V_H(m,k) } \cdot \sum_{i \in [n]} H_i T_i^* \ . \]
As stated in Claim~\ref{clm:optimal_base_shifted}, whose complete derivation is deferred to Appendix~\ref{app:proof_clm_optimal_base_shifted}, the optimal choice here is $m^* = 3$ and $k^* = 2$, attaining an expected operational cost within factor $\frac{ 2(\sqrt{3}-1) }{ \ln 3 } < 1.3327$ of optimal. This result represents a substantial improvement over the deterministic $1.3776$-approximation ratio of Section~\ref{sec:unshifted}.

\begin{claim} \label{clm:optimal_base_shifted}
The integers $m \in \bbZ_{\geq 2}$ and $k \in \bbZ_{\geq 1}$ that minimize $\max \{ V_J(m,k), V_H(m,k) \}$ are $m^* = 3$ and $k^* = 2$, yielding an optimum value of $\frac{ 2(\sqrt{3}-1) }{ \ln 3 }$.
\end{claim}

\begin{lemma} \label{lem:approx_rand_3_2}
$\exsubpar{ \Theta }{ F( T^{3,2, \Theta} ) } \leq \frac{ 2(\sqrt{3}-1) }{ \ln 3 } \cdot \opt\eqref{eqn:conv_relax_RCJRP}$.
\end{lemma}

\paragraph{Best-of-both-worlds policy.} Our primary improvement along these lines is attained by exploiting a structural asymmetry in the randomized policy $T^{2,1,\Theta}$. As previously mentioned, by instantiating the rational-case bound~\eqref{eqn:bound_rand_rational} on the expected operational cost of $T^{2,1, \Theta}$, we have 
\begin{equation} \label{eqn:def_FT_rand_21}
\exsub{ \Theta }{ F(T^{2,1, \Theta}) } ~~\leq~~ \frac{ 1 }{ 2 \ln 2 } \cdot \frac{ K_0 }{ T_{\min}^* } + \frac{ 1 }{ 2 \ln 2 } \cdot \sum_{i \in [n]} \frac{ K_i }{ T_i^* } + \frac{ 1 }{ \ln 2 } \cdot \sum_{i \in [n]} H_i T_i^* \ .
\end{equation}
Interestingly, this policy's ordering density multiplier is only $\frac{ 1 }{ 2 \ln 2 } \approx 0.7213$, even though it incurs heavy holding costs, with a multiplier of $\frac{ 1 }{ \ln 2 } \approx 1.4426$. To counterbalance this trend, we evaluate $T^{2,2, \Theta}$, which exhibits the opposite asymmetry. Indeed, 
\begin{equation} \label{eqn:def_FT_rand_22}
\exsub{ \Theta }{ F(T^{2,2, \Theta}) } ~~\leq~~ \frac{ 1 }{ \ln 2 } \cdot \frac{ K_0 }{ T_{\min}^* } + \frac{ 2-\sqrt{2} }{ \ln 2 } \cdot \sum_{i \in [n]} \frac{ K_i }{ T_i^* } + \frac{ 2(\sqrt{2}-1) }{ \ln 2 } \cdot \sum_{i \in [n]} H_i T_i^* \ ,
\end{equation}
which is associated with a high ordering density multiplier of $\frac{ 1 }{ \ln 2 } \approx 1.4426$, but with a highly efficient holding cost multiplier, $\frac{ 2(\sqrt{2}-1) }{ \ln 2 } \approx 1.1951$. By picking the cheaper of these two randomized policies, we push our expected performance guarantee down to $1.2585$. 

\begin{lemma} \label{lem:approx_better_2rand}
$\exsubpar{ \Theta }{ \min \{ F( T^{2,1, \Theta} ), F( T^{2,2, \Theta} ) \} } \leq 1.2585 \cdot \opt\eqref{eqn:conv_relax_RCJRP} $.
\end{lemma}
\begin{proof}
By upper-bounding the expected operational costs of $T^{2,1, \Theta}$ and $T^{2,2, \Theta}$ via inequalities~\eqref{eqn:def_FT_rand_21} and~\eqref{eqn:def_FT_rand_22}, we conclude that
\begin{align*}
& \exsub{ \Theta }{ \min \{ F( T^{2,1, \Theta} ), F( T^{2,2, \Theta} ) \} } \\
& \qquad \leq~~ \min_{\lambda \in [0,1]} \left\{ \lambda \cdot \exsub{ \Theta }{ F( T^{2,1, \Theta} ) } + (1-\lambda) \cdot \exsub{ \Theta }{ F( T^{2,2, \Theta} ) } \right\} \\
& \qquad \leq~~ \min_{\lambda \in [0,1]} \left\{ \lambda \cdot \left( \frac{ 1 }{ 2 \ln 2 } \cdot \frac{ K_0 }{ T_{\min}^* } + \frac{ 1 }{ 2 \ln 2 } \cdot \sum_{i \in [n]} \frac{ K_i }{ T_i^* } + \frac{ 1 }{ \ln 2 } \cdot \sum_{i \in [n]} H_i T_i^* \right) \right. \\
& \qquad \qquad \qquad \qquad \left. \mbox{} + (1-\lambda) \cdot \left( \frac{ 1 }{ \ln 2 } \cdot \frac{ K_0 }{ T_{\min}^* } + \frac{ 2-\sqrt{2} }{ \ln 2 } \cdot \sum_{i \in [n]} \frac{ K_i }{ T_i^* } + \frac{ 2(\sqrt{2}-1) }{ \ln 2 } \cdot \sum_{i \in [n]} H_i T_i^* \right) \right\} \\
&\qquad \leq~~ 0.25548 \cdot \left( \frac{ 1 }{ 2 \ln 2 } \cdot \frac{ K_0 }{ T_{\min}^* } + \frac{ 1 }{ 2 \ln 2 } \cdot \sum_{i \in [n]} \frac{ K_i }{ T_i^* } + \frac{ 1 }{ \ln 2 } \cdot \sum_{i \in [n]} H_i T_i^* \right) \\
& \qquad \qquad \mbox{} + (1-0.25548) \cdot \left( \frac{ 1 }{ \ln 2 } \cdot \frac{ K_0 }{ T_{\min}^* } + \frac{ 2-\sqrt{2} }{ \ln 2 } \cdot \sum_{i \in [n]} \frac{ K_i }{ T_i^* } + \frac{ 2(\sqrt{2}-1) }{ \ln 2 } \cdot \sum_{i \in [n]} H_i T_i^* \right) \\
& \qquad \leq~~ 1.2585 \cdot \left( \frac{ K_0 }{ T_{\min}^* } + \sum_{i \in [n]} \left( \frac{ K_i }{ T_i^* } + H_i T_i^* \right) \right) \\
& \qquad =~~ 1.2585 \cdot \opt\eqref{eqn:conv_relax_RCJRP} \ . \nonumber
\end{align*}
\end{proof}

\subsection{Best-of-two policy via factor-revealing LP} \label{subsec:factor_reveal}

Despite extensive numerical search, we still do not know whether the best-of-several-worlds approach can improve upon $T^{2,1,\Theta}$ and $T^{2,2,\Theta}$, in the sense of achieving better expected cost coefficients via two or more randomized policies. In contrast, our current approach treats the shift $\Theta$ as an intentionally selected design parameter. This idea motivates us to exploit $\tilde{T}^{2,1,\Theta}$ and $\tilde{T}^{2,2,\Theta}$, along with an instance-dependent choice of $\Theta$. Here, the $\tilde{T}$-notation indicates that we will not be directly rounding the optimal convex relaxation solution $T^*$, but rather a carefully modified vector. As stated below, this proactive choice allows us to surpass the approximation guarantee stated in Lemma~\ref{lem:approx_better_2rand}, which evaluates to at least $1.2584$.

\begin{lemma} 
$\min_{ \Theta \in [0,1] } \{ F( \tilde{T}^{2,1, \Theta} ), F( \tilde{T}^{2,2, \Theta} ) \} \leq 1.2512 \cdot \opt\eqref{eqn:conv_relax_RCJRP}$.
\end{lemma}

\paragraph{Policy description.} Rather than working with an optimal solution $T^*$ to the convex relaxation~\eqref{eqn:conv_relax_RCJRP}, our policies will be defined in terms of its $\tilde{T}$-counterpart. Given an integer-valued discretization parameter $N \geq 1$, to be specified later on, for every commodity $i \in [n]$, we fix $\tilde{T}_i$ by rounding $T^*_i$ up to the nearest integer power of $2^{1/N}$; similarly, $T_{\min}^*$ is transformed to $\tilde{T}_{\min}$. As such, $\tilde{T}$ is a feasible solution to~\eqref{eqn:conv_relax_RCJRP}, with an objective value of at most $2^{1/N} \cdot \opt\eqref{eqn:conv_relax_RCJRP}$. With this modified vector, the policies $\tilde{T}^{2,1,\Theta}$ and $\tilde{T}^{2,2,\Theta}$ are defined in relation to $\tilde{T}$, meaning that $\tilde{T}^{m,k, \Theta}_i = \myround_{m,k, \Theta}( \tilde{T}_i )$ for every commodity $i \in [n]$. Subsequently, we pick the shift $\Theta \in [0,1]$ for which $\min \{ F( \tilde{T}^{2, 1, \Theta} ), F( \tilde{T}^{2,2, \Theta} ) \}$ is minimized. As explained in Section~\ref{sec:conclusions} when discussing derandomization, one can indeed determine an optimal choice of $\Theta$ for all policies considered in this paper. 

\paragraph{Bounding $F(\tilde{T}^{2,k, \Theta})$ for adversarial instances.} Let us recall that the objective value of $\tilde{T}$ with respect to formulation~\eqref{eqn:conv_relax_RCJRP} can be written as $\Psi( \tilde{T} ) = \frac{ K_0 }{ \tilde{T}_{\min} } + \sum_{i \in [n]} \frac{ K_i }{ \tilde{T}_i } + \sum_{i \in [n]} H_i \tilde{T}_i$. For the purpose of decomposing this expression, note that the restriction of each time interval $\tilde{T}_i$ to integer powers of $2^{1/N}$ induces a natural partition of the underlying commodities into $N$ disjoint sets. Specifically, for each residue class $\nu \in [N]$, we define $A_{\nu}$ as the subset of commodities $i \in [n]$ whose modified time interval take the form $\tilde{T}_i = 2^{r_i + \nu/N}$ for some integer $r_i$. Consequently, 
\[ \Psi( \tilde{T} ) ~~=~~ \underbrace{ \vphantom{\sum_{i \in A_{\nu}} H_i \tilde{T}_i} \frac{ K_0 }{ \tilde{T}_{\min} } }_{ z } + \sum_{ \nu \in [N] } \underbrace{ \sum_{i \in A_{\nu}} \frac{ K_i }{ \tilde{T}_i } }_{ x_{\nu} } + \sum_{ \nu \in [N] } \underbrace{ \sum_{i \in A_{\nu}} H_i \tilde{T}_i }_{ y_{\nu} } \ . \]

The crucial observation is that, with this shorthand notation, for any choice of $\Theta$, the operational cost of $\tilde{T}^{2,k, \Theta}$ can be bounded in terms of $z$, $\{ x_{\nu} \}_{\nu \in [N]}$, and $\{ y_{\nu} \}_{\nu \in [N]}$. Indeed, we know that this policy incurs a joint ordering cost of $J( \tilde{T}^{2,k, \Theta} ) \leq K_0 \cdot \tilde{\cal D}_{2,k,\Theta}$, where the latter coefficient is given by 
\[ {\cal D}_{2,k, \Theta} ~~=~~ \frac{ 1 }{ \myround_{2,k,\Theta}( \tilde{T}_{\min} ) } \cdot \sum_{\emptyset \neq{\cal N} \subseteq {\cal S}_{2,k}} \frac{ (-1)^{ |{\cal N}| + 1} }{ \mylcm({\cal N}) } ~~=~~ \frac{ 1 }{ \tilde{T}_{\min} \cdot 2^{\Theta/k} } \cdot \sum_{\emptyset \neq{\cal N} \subseteq {\cal S}_{2,k}} \frac{ (-1)^{ |{\cal N}| + 1} }{ \mylcm({\cal N}) } \ , \]
as explained within the proof of Claim~\ref{clm:prop_shifted}. In addition, for every class $\nu \in [N]$ and commodity $i \in A_{\nu}$, its time interval $\tilde{T}_i = 2^{r_i + \nu/N}$ has a rounding error of 
\[ \frac{ \tilde{T}^{2,k, \Theta}_i }{ \tilde{T}_i } ~~=~~ \frac{ \myround_{2,k, \Theta}( \tilde{T}_i ) }{ \tilde{T}_i } ~~=~~ \frac{ \myround_{2,k, \Theta}( 2^{r_i + \nu/N} ) }{ 2^{r_i + \nu/N} } ~~=~~ \frac{ \myround_{2,k, \Theta}( 2^{\nu/N} ) }{ 2^{\nu/N} } \ . \]
In other words, we incur precisely the same error for all commodities in $A_{\nu}$, noting that $\myround_{2,k, \Theta}( 2^{\nu/N} )$ can easily be computed; see equation~\eqref{eqn:def_rounding_mktheta} in Section~\ref{subsec:proof_clm_prop_shifted}. By aggregating these observations, we have
\[ F(\tilde{T}^{2,k, \Theta}) ~~\leq~~ 2^{-\Theta/k} \cdot \left( \sum_{\emptyset \neq{\cal N} \subseteq {\cal S}_{2,k}} \frac{ (-1)^{ |{\cal N}| + 1} }{ \mylcm({\cal N}) } \right) \cdot z + \sum_{ \nu \in [N] } \frac{ 2^{\nu/N} }{ \myround_{2,k, \Theta}( 2^{\nu/N} ) } \cdot x_{\nu} + \sum_{ \nu \in [N] } \frac{ \myround_{2,k, \Theta}( 2^{\nu/N} ) }{ 2^{\nu/N} } \cdot y_{\nu} \ . \]

\paragraph{The factor-revealing LP.} To evaluate the best-performing policy out of $\tilde{T}^{2,1,\Theta}$ and $\tilde{T}^{2,2,\Theta}$ across all continuous choices of $\Theta \in [0,1]$, we transition to a tractable framework by discretizing this interval, restricting our attention to the set $\Theta_{\mydisc} = \{ \frac{ \ell }{ L }: 0 \leq \ell \leq L \}$. Clearly, for any integer-valued $L \geq 1$, the resulting approximation guarantee will form an upper bound on the one attained by optimizing over all $\Theta \in [0,1]$. The main advantage of considering finitely many $\Theta$ values is that such discretization allows us to rigorously derive an upper bound via a factor-revealing LP. The goal of this mechanism is to show that, regardless of how our adversary sets up $z$, $\{ x_{\nu} \}_{\nu \in [N]}$, and $\{ y_{\nu} \}_{\nu \in [N]}$, there is always a choice of $\Theta \in \Theta_{\mydisc}$ ensuring an approximation ratio of $\rho$, which is what we wish to optimize. For this purpose, when these adversarial terms are normalized by the objective value $\Psi( \tilde{T} )$, we obtain the following linear program:
\begin{alignat}{3} 
& \text{max} \quad && \rho \tag{$\mathrm{LP}_{N,L}$} \\
&\text{s.t.} \quad && \rho \leq 2^{-\Theta} \cdot \left( \sum_{\emptyset \neq{\cal N} \subseteq {\cal S}_{2,1}} \frac{ (-1)^{ |{\cal N}| + 1} }{ \mylcm({\cal N}) } \right) \cdot z \nonumber \\
&&& \qquad \quad \mbox{} + \sum_{ \nu \in [N] } \frac{ 2^{\nu/N} }{ \myround_{2,1, \Theta}( 2^{\nu/N} ) } \cdot x_{\nu} + \sum_{ \nu \in [N] } \frac{ \myround_{2,1, \Theta}( 2^{\nu/N} ) }{ 2^{\nu/N} } \cdot y_{\nu} \qquad && \forall \, \Theta \in \Theta_{\mydisc} \nonumber \\
& \quad && \rho \leq 2^{-\Theta/2} \cdot \left( \sum_{\emptyset \neq{\cal N} \subseteq {\cal S}_{2,2}} \frac{ (-1)^{ |{\cal N}| + 1} }{ \mylcm({\cal N}) } \right) \cdot z \nonumber \\
&&& \qquad \quad \mbox{} + \sum_{ \nu \in [N] } \frac{ 2^{\nu/N} }{ \myround_{2,2, \Theta}( 2^{\nu/N} ) } \cdot x_{\nu} + \sum_{ \nu \in [N] } \frac{ \myround_{2,2, \Theta}( 2^{\nu/N} ) }{ 2^{\nu/N} } \cdot y_{\nu} \qquad && \forall \, \Theta \in \Theta_{\mydisc} \nonumber \\
&&& z + \sum_{\nu \in [N]} x_{\nu} + \sum_{\nu \in [N]} y_{\nu} = 1 \nonumber \\
&&& z, \, x_{\nu}, \, y_{\nu} \geq 0 && \forall \, \nu \in [N] \nonumber 
\end{alignat}
By numerically solving this program, instantiated with $N=2000$ and $L=2000$, we find that $\opt(\mathrm{LP}_{2000,2000}) = 1.250677$. That said, one should still take into account that, due to normalizing by $\Psi( \tilde{T} )$, this worst-case ratio is determined with respect to the modified solution $\tilde{T}$. As previously mentioned, the ratio between $\Psi( \tilde{T} )$ and $\Psi( T^* )$ is bounded by $2^{1/N}$, implying that our final approximation guarantee is at most $1.250677 \cdot 2^{1/2000} < 1.2512$.

\subsection{Proof of Claim~\ref{clm:prop_shifted}} \label{subsec:proof_clm_prop_shifted}

\paragraph{Proof of items~\ref{item:prop_shifted_1} and~\ref{item:prop_shifted_2}.} 
To evaluate the expected values of $\myround_{m,k, \Theta}(t)$ and its reciprocal, we define their relationship to $t$ as a function of the random shift $\Theta$. For this purpose, let $B = m^{1/k}$ be the base expansion factor and let $L_t = \log_B(\frac{ t }{ T_{\min}^*} )$ be a mapping of $t > 0$ to the logarithmic scale, normalized by $T_{\min}^*$. Then, in terms of $\Theta$, the rounded value of $t$ is given by the relative position of $L_t$ within a unit cycle, along the following case disjunction:
\begin{equation} \label{eqn:def_rounding_mktheta}
\myround_{m,k, \Theta}(t) ~~=~~ 
\begin{cases} 
t \cdot B^{\lfloor L_t \rfloor - L_t + 1 + \Theta}, \qquad & \text{when } \Theta \in [0, L_t - \lfloor L_t \rfloor] \\
t \cdot B^{\lfloor L_t \rfloor - L_t + \Theta}, & \text{when } \Theta \in (L_t - \lfloor L_t \rfloor, 1]
\end{cases} 
\end{equation}

To compute $\exsubpar{ \Theta }{ \myround_{m,k, \Theta}(t) }$, we integrate this functional dependency over $\Theta \in [0,1]$, where for simplicity of notation, $\phi = L_t - \lfloor L_t \rfloor$, leading to
\begin{align*}
\exsub{ \Theta }{ \myround_{m,k, \Theta}(t) } ~~=~~ & t \cdot \left( \int_{0}^{\phi} B^{1-\phi + \theta} \mathrm{d}\theta + \int_{\phi}^{1} B^{\theta - \phi} \mathrm{d}\theta \right) \\
~~=~~& t \cdot \left( B^{1-\phi} \cdot \left. \frac{B^\theta}{\ln B} \right]_{0}^{\phi} + B^{-\phi} \cdot \left. \frac{B^\theta}{\ln B} \right]_{\phi}^{1} \right) \\
~~=~~& \frac{t}{\ln B} \cdot \left( B^{1-\phi} \cdot (B^\phi - 1) + B^{-\phi} \cdot (B - B^\phi) \right) \\
~~=~~& \frac{B-1}{\ln B} \cdot t \ .
\end{align*}
Substituting back $B = m^{1/k}$, we obtain $\exsubpar{ \Theta }{ \myround_{m,k, \Theta}(t) } = \frac{ k(m^{1/k}-1) }{ \ln m } \cdot t$, completing the proof of item~\ref{item:prop_shifted_1}. For the reciprocal, a similar integration yields
\begin{align*}
\exsub{ \Theta }{ \frac{1}{\myround_{m,k, \Theta}(t)} } ~~=~~& \frac{ 1 }{ t } \cdot \left( \int_{0}^{\phi} \frac{\mathrm{d}\theta}{B^{1 - \phi + \theta}} + \int_{\phi}^{1} \frac{\mathrm{d}\theta}{B^{\theta - \phi}} \right) \\
~~=~~& \frac{1}{t} \cdot \left( \left. \frac{-B^{\phi - 1 - \theta}}{\ln B} \right]_{0}^{\phi} + \left. \frac{-B^{\phi - \theta}}{\ln B} \right]_{\phi}^{1} \right) \\
~~=~~& \frac{1}{t \ln B} \cdot \left( (B^{\phi-1} - B^{-1}) + (1 - B^{\phi-1}) \right) \\
~~=~~& \frac{B-1}{B \ln B} \cdot \frac{ 1 }{ t } \ .
\end{align*}
Again, substituting $B = m^{1/k}$ results in $\exsubpar{ \Theta }{ \frac{ 1 }{ \myround_{m,k, \Theta}(t) } } = \frac{ k(m^{1/k}-1) }{ m^{1/k} \ln m } \cdot \frac{1}{t}$, which concludes the proof of item~\ref{item:prop_shifted_2}.

\paragraph{Proof of item~\ref{item:prop_shifted_3}.} By replicating the analysis leading to Claim~\ref{clm:exact_density} in the deterministic setting, it follows that ${\cal D}_{m,k, \Theta} = \frac{ 1 }{ \myround_{m,k,\Theta}( T_{\min}^* ) } \cdot \sum_{\emptyset \neq{\cal N} \subseteq {\cal S}_{m,k}} \frac{ (-1)^{ |{\cal N}| + 1} }{ \mylcm({\cal N}) }$. Therefore,
\begin{align*}
\exsub{ \Theta }{ {\cal D}_{m,k, \Theta} } ~~=~~ & \exsub{ \Theta }{ \frac{ 1 }{ \myround_{m,k, \Theta}( T_{\min}^* ) } } \cdot \sum_{\emptyset \neq{\cal N} \subseteq {\cal S}_{m,k}} \frac{ (-1)^{ |{\cal N}| + 1} }{ \mylcm({\cal N}) } \\
~~=~~ & \frac{ k(m^{1/k}-1)}{ m^{1/k}\ln m} \cdot \left( \sum_{\emptyset \neq{\cal N} \subseteq {\cal S}_{m,k}} \frac{ (-1)^{ |{\cal N}| + 1} }{ \mylcm({\cal N}) } \right) \cdot \frac{ 1 }{ T_{\min}^* } \ , 
\end{align*}
where the second equality follows from the expression in item~\ref{item:prop_shifted_2} for the expected value of $\frac{ 1 }{ \myround_{m,k, \Theta}(t) }$. 
\section{\texorpdfstring{Shifted Interleaved Power-of-$m^{1/k}$ Policies}{}} \label{sec:Interleaved}

In this section, we finalize our algorithmic machinery by introducing the class of shifted interleaved power-of-$m^{1/k}$ policies. This approach utilizes a secondary offset grid to subdivide the rounding gaps of a primary geometric grid, thereby mitigating its expected holding cost inflation. By optimizing this offset, we demonstrate that a single-shot interleaved policy strikes an effective balance between joint ordering and holding cost inflations. These ideas allow us to establish our main result, showing that the resource-constrained joint replenishment problem can be efficiently approximated within factor $\frac{5}{6\ln 2} < 1.2023$ of optimal. Moreover, we prove that our choice of parameters perfectly balances between various cost components, arguing that this configuration is optimal for the entire class of interleaved policies.

\subsection{Introducing shifted interleaved policies} \label{subsec:rand_interleaved_prelim}

The core logic behind shifted interleaved policies focuses on mitigating the expected rounding errors inherent to a single geometric grid. Specifically, with an expansion factor of $m^{1/k}$, any given time interval $T_i^*$ might be inflated by nearly $m^{1/k}$ when rounded up. Through interleaving a secondary geometric grid, offset by a carefully chosen factor $\alpha$, we divide each grid interval into two smaller sub-intervals. This idea reduces the maximum rounding error of any point, tightening the expected individual ordering and holding cost multipliers. Even though an additional grid comes at the expense of a denser set of joint ordering points, we will prove that specific choices for the offset $\alpha$ strike an effective balance between these competing costs. Toward this objective, our interleaved geometric grid, rounding operator, and density coefficient are formally defined as follows.

\paragraph{The interleaved grid.} Let $m \geq 2$ and $k \geq 1$ be a pair of integers whose values will be optimized later on, dictating our base expansion factor, $m^{1/k}$. We begin by considering the geometric grid ${\cal G}_{m,k} = \{ m^{p/k} \cdot T_{\min}^* : p \in \bbZ \}$, where $T_{\min}^*$ is multiplied by all integer powers of $m^{1/k}$. Additionally, for a constant parameter $\alpha \in (1, m^{1/k})$ whose value will be optimized as well, we define the $\alpha$-scaled secondary grid ${\cal G}_{m,k}^{\alpha} = \{ \alpha m^{p/k} \cdot T_{\min}^* : p \in \bbZ \}$. The union of these two grids forms the interleaved grid ${\cal H}_{m,k}^{\alpha} = {\cal G}_{m,k} \cup {\cal G}_{m,k}^{\alpha}$. Moving forward, rather than directly utilizing ${\cal H}_{m,k}^{\alpha}$, we scale its elements by a multiplicative factor of $m^{\Theta/k}$, where $\Theta \sim U[0,1]$ represents a uniformly distributed shift, thereby creating the randomized interleaved grid ${\cal H}_{m,k, \Theta}^{ \alpha } = \{ m^{\Theta/k} \cdot g : g \in {\cal H}_{m,k}^{\alpha} \}$.

\paragraph{The operator $\myround^{\alpha}_{m,k, \Theta}$ and density coefficient ${\cal D}_{m,k,\Theta}^{ \alpha }$.} The first implication of these definitions is that, when any $t > 0$ is rounded up to the strictly greater nearest point in ${\cal H}_{m,k, \Theta}^{ \alpha }$, the resulting point is random. As such, our rounding operator $\myround^{\alpha}_{m,k, \Theta}( t ) = \min \{ g \in {\cal H}_{m,k, \Theta}^{ \alpha } : g > t \}$ is actually a random function. Similarly, the asymptotic density of the relevant grid multiples is random as well, with the density coefficient of ${\cal H}_{m,k, \Theta}^{ \alpha }$ specified by
\[ {\cal D}_{m,k,\Theta}^{ \alpha } ~~=~~ \lim_{\Delta \to \infty} \frac{ | \bigcup_{g \in {\cal H}_{m,k, \Theta}^{ \alpha } \cap [\myround^{\alpha}_{m,k, \Theta}( T_{\min}^* ),\infty) } {\cal M}_{g,\Delta} | }{ \Delta } \ . \]

\paragraph{Creating $T^{m,k, \alpha, \Theta}$.} Given an optimal convex relaxation solution $T^* = (T_{\min}^*, T_1^*, \ldots, T_n^*)$, the randomized interleaved policy $T^{m,k, \alpha, \Theta}$ is defined via our new rounding operator. Namely, we set $T^{m,k, \alpha, \Theta}_i = \myround^{\alpha}_{m,k, \Theta}( T_i^* )$ for every commodity $i \in [n]$. As before, since this operator strictly increases its argument, the resulting policy remains resource-feasible with probability $1$.

\paragraph{Parametric performance guarantee.} Next, we explain how to determine the expected values of the rounding operator $\myround^\alpha_{m,k, \Theta}$ and its reciprocal, which govern the individual ordering and holding costs of each commodity. As stated below, these values admit closed-form expressions, including dependencies on how the offset $\alpha$ splits the logarithmic interval of the base grid. Furthermore, we derive an analogous expression with respect to the expected density coefficient ${\cal D}_{m,k,\Theta}^{ \alpha }$. In this context, by extending the notation of Section~\ref{subsec:det_rounding}, we make use of ${\cal S}_{m,k}^{ \alpha, \myprim } = \{ m^{(\kappa-1)/k} \}_{\kappa \in [k]} \cup \{ \alpha m^{(\kappa-1)/k} \}_{\kappa \in [k]}$ to designate the primary base multipliers of the interleaved grid. Similarly, ${\cal S}_{m,k}^{ \alpha, \mysecond } = \{ m^{(\kappa-1)/k} \}_{\kappa \in [k]} \cup \{ \frac{1}{\alpha} m^{\kappa/k} \}_{\kappa \in [k]}$ will denote its secondary base multipliers. For convenience, their corresponding inclusion-exclusion summations will be denoted by $\Sigma_{m,k}^{ \alpha, \myprim} = \sum_{\emptyset \neq {\cal N} \subseteq {\cal S}_{m,k}^{ \alpha, \myprim}} \frac{ (-1)^{ |{\cal N}| + 1} }{ \mylcm({\cal N}) }$ and $\Sigma_{m,k}^{ \alpha, \mysecond} = \sum_{\emptyset \neq {\cal N} \subseteq {\cal S}_{m,k}^{ \alpha, \mysecond}} \frac{ (-1)^{ |{\cal N}| + 1} }{ \mylcm({\cal N}) }$. The proof of this result is presented in Section~\ref{subsec:proof_clm_prop_interleaved_rounding_final}; as a sanity check, one can verify that Claim~\ref{clm:prop_shifted} is recovered by setting $\alpha = 1$.

\begin{claim} \label{clm:prop_interleaved_rounding_final}
\begin{enumerate}
    \item $\exsubpar{ \Theta }{ \myround^\alpha_{m,k, \Theta}(t) } = \frac{ k }{ \ln m } \cdot ( \alpha + \frac{ m^{1/k} }{\alpha} - 2 ) \cdot t$, for every $t > 0$. \label{item:final_prop_1}

    \item $\exsubpar{ \Theta }{ \frac{ 1 }{ \myround^\alpha_{m,k, \Theta}(t) } } = \frac{ k }{ \ln m } \cdot ( 2 - \frac{1}{\alpha} - \frac{\alpha}{m^{1/k}} ) \cdot \frac{ 1 }{ t } $, for every $t > 0$. \label{item:final_prop_2} 

    \item $\exsubpar{ \Theta }{ {\cal D}_{m,k,\Theta}^{ \alpha } } = \frac{ k }{ \ln m } \cdot ( ( 1 - \frac{\alpha}{ m^{1/k} } ) \cdot \Sigma_{m,k}^{ \alpha, \myprim} + ( 1 - \frac{1}{\alpha} ) \cdot \Sigma_{m,k}^{ \alpha, \mysecond} ) \cdot \frac{ 1 }{ T_{\min}^* } $. \label{item:final_prop_3}
    \end{enumerate}
\end{claim}

Given these findings, the expected operational cost of $T^{m,k, \alpha, \Theta}$ can be related to the convex relaxation optimum in terms of the grid parameters $m$, $k$, and $\alpha$, by noting that
\begin{align}
& \exsub{ \Theta }{ F(T^{m,k, \alpha, \Theta}) } ~~\leq~~ K_0 \cdot \exsub{ \Theta }{ {\cal D}_{m,k, \Theta}^{ \alpha }} + \sum_{i \in [n]} \left( K_i \cdot \exsub{ \Theta }{ \frac{ 1 }{ T^{m,k, \alpha, \Theta}_i } } + H_i \cdot \exsub{ \Theta }{ T^{m,k, \alpha, \Theta}_i } \right) \nonumber \\
& \qquad =~~ K_0 \cdot \exsub{ \Theta }{ {\cal D}_{m,k, \Theta}^{ \alpha }} + \sum_{i \in [n]} \left( K_i \cdot \exsub{ \Theta }{ \frac{ 1 }{ \myround^{\alpha}_{m,k, \Theta}( T_i^* ) } } + H_i \cdot \exsub{ \Theta }{ \myround^{\alpha}_{m,k, \Theta}( T_i^* ) } \right) \nonumber \\
& \qquad \leq~~ \underbrace{ \frac{ k }{ \ln m } \cdot \left( \left( 1 - \frac{\alpha}{ m^{1/k} } \right) \cdot \Sigma_{m,k}^{ \alpha, \myprim} + \left( 1 - \frac{1}{\alpha} \right) \cdot \Sigma_{m,k}^{ \alpha, \mysecond} \right) }_{ V_J(m,k,\alpha) }\cdot \frac{ K_0 }{ T_{\min}^* } \nonumber \\
& \qquad \qquad \mbox{} + \underbrace{ \vphantom{\frac{ k }{ \ln m } \cdot \left( \alpha + \frac{ m^{1/k} }{\alpha} - 2 \right)} \frac{ k }{ \ln m } \cdot \left( 2 - \frac{1}{\alpha} - \frac{\alpha}{m^{1/k}} \right) }_{ \leq 1 } \cdot \sum_{i \in [n]} \frac{ K_i }{ T_i^* } + \underbrace{ \frac{ k }{ \ln m } \cdot \left( \alpha + \frac{ m^{1/k} }{\alpha} - 2 \right) }_{ V_H(m,k,\alpha) } \cdot \sum_{i \in [n]} H_i T_i^* \ . \label{eqn:bound_interleaved_shifted}
\end{align}

\subsection{Policy instantiation} \label{subsec:interleave_inst}

\paragraph{Choice of parameters.} From this point on, we will investigate the specific configuration where $m = 2$ and $k = 1$, along with an offset of $\alpha = \frac{ 3 }{ 2 }$. In other words, our randomized policy will be $T^{2, 1, 3/2, \Theta}$, with the convention that $T^{2, 1, 3/2, \Theta}_i = \myround^{3/2}_{2,1, \Theta}( T_i^* )$ for every commodity $i \in [n]$. Intuitively, the latter offset is small enough to significantly subdivide the primary grid, thereby reducing holding cost inflation. Furthermore, by choosing $\alpha = \frac{ 3 }{ 2 }$, we ensure that the secondary grid points are well-staggered relative to the primary ones, leading to useful synchronization. Specifically, since the current expansion factor and offset are both simple rationals, integer multiples of these grids frequently overlap, preventing the joint ordering density from increasing as sharply as the sum of its parts. Notably, in Section~\ref{subsec:optimal_conf_inter}, we show that the choice of $m = 2$, $k = 1$, and $\alpha = \frac{ 3 }{ 2 }$ is best possible for the class of interleaved power-of-$m^{1/k}$ policies.

\paragraph{Performance guarantee.} To quantify the expected operational cost of $T^{2, 1, 3/2, \Theta}$, let us consider the three coefficients appearing in our upper bound~\eqref{eqn:bound_interleaved_shifted}. Beyond directly plugging in $m=2$, $k=1$, and $\alpha = \frac{ 3 }{ 2 }$, it remains to evaluate the inclusion-exclusion summation over the primary base multipliers ${\cal S}_{2,1}^{ 3/2, \myprim} = \{ 1, \frac{ 3 }{ 2 } \}$, which amounts to $\Sigma_{2,1}^{ 3/2, \myprim} = 1 + \frac{ 2 }{ 3 } - \frac{ 1 }{ \mylcm(1, 3/2) } = \frac{ 4 }{ 3 }$. Similarly, in regard to the secondary base multipliers ${\cal S}_{2,1}^{ 3/2, \mysecond} = \{ 1, \frac{ 4 }{ 3 } \}$, we have $\Sigma_{2,1}^{ 3/2, \mysecond} = 1 + \frac{ 3 }{ 4 } - \frac{ 1 }{ \mylcm(1, 4/3) } = \frac{ 3 }{ 2 }$. Consequently, we arrive at an upper bound of 
\begin{equation} \label{eqn:UB_interleaved_2_1_1.5}
\exsub{ \Theta }{ F(T^{2, 1, 3/2, \Theta}) } ~~\leq~~ \frac{ 5 }{ 6\ln 2 } \cdot \frac{ K_0 }{ T_{\min}^* } + \frac{ 7 }{ 12\ln 2 } \cdot \sum_{i \in [n]} \frac{ K_i }{ T_i^* } + \frac{ 5 }{ 6\ln 2 } \cdot \sum_{i \in [n]} H_i T_i^* 
\end{equation}
on the expected operational cost of $T^{2,1, 3/2, \Theta}$, which is within factor $\frac{ 5 }{ 6\ln 2 } < 1.2023$ of the convex relaxation optimum. This result constitutes a significant improvement over the best-of-two randomized policies of Section~\ref{subsec:factor_reveal}, which provided an approximation ratio of $1.2512$.

\begin{lemma} \label{lem:approx_interleaved_specialized}
$\exsubpar{ \Theta }{ F( T^{2,1, 3/2, \Theta} ) } \leq \frac{ 5 }{ 6 \ln 2 } \cdot \opt\eqref{eqn:conv_relax_RCJRP}$.
\end{lemma}

\subsection{Optimality of single-shot configuration} \label{subsec:optimal_conf_inter}

In what follows, we prove that the configuration considered in Section~\ref{subsec:interleave_inst} is optimal, in the sense of minimizing the largest coefficient we are seeing in our general-purpose bound~\eqref{eqn:bound_interleaved_shifted}. Interestingly, additional optimal configurations exist; for instance, the upcoming analysis will demonstrate that $m = 2$, $k = 1$, and $\alpha = \frac{ 4 }{ 3 }$ is best possible as well.

\begin{lemma} \label{lem:optimal_configuration}
A configuration of integers $m \in \bbZ_{\geq 2}$ and $k \in \bbZ_{\geq 1}$, along with an offset $\alpha \in (1, m^{1/k})$, that minimizes $\max \{ V_J(m,k,\alpha), V_H(m,k,\alpha) \}$ is given by $m^* = 2$, $k^* = 1$, and $\alpha^* = \frac{ 3 }{ 2 }$.
\end{lemma}

As previously explained, by plugging the above-mentioned configuration into our density and holding cost coefficients, one can ascertain that $V_J(2,1,\frac{3}{2}) = V_H(2,1,\frac{3}{2}) = \frac{ 5 }{ 6 \ln 2}$. Moving forward, we show that in any other case, $\max \{ V_J(m,k,\alpha), V_H(m,k,\alpha) \} \geq \frac{ 5 }{ 6 \ln 2}$. 

\paragraph{Holding coefficient constraints.} To this end, letting $B = m^{1/k}$ be our base expansion factor, we observe that the holding cost coefficient $V_H(B,\alpha) = \frac{1}{\ln B} \cdot ( \alpha + \frac{B}{\alpha} - 2 )$ is strictly convex in $\alpha$, and attains its global minimum over $(0,\infty)$ at $\alpha = \sqrt{B}$. Therefore, for any configuration, $V_H(B,\alpha) \geq \frac{2(\sqrt{B} - 1)}{\ln B}$. Noting that the function $B \mapsto \frac{2(\sqrt{B} - 1)}{\ln B}$ is strictly increasing over $(1,\infty)$, it follows that $V_H(B,\alpha) \geq \frac{2(\sqrt{41/20} - 1)}{\ln (41/20)} > \frac{5}{6 \ln 2}$ for every $B \geq \frac{ 41 }{ 20 }$. As such, optimal configurations must pick $B < \frac{ 41 }{ 20 }$. In turn, we observe that such configurations cannot select an offset of the form $\alpha = \frac{ B }{ q }$, for any integer $q \geq 1$. Indeed, since $\alpha \in (1,B)$ and $B < \frac{ 41 }{ 20 }$, both $q=1$ and $q \geq 3$ are infeasible choices. The remaining case of $q \geq 2$ cannot be optimal, as our holding cost coefficient would become 
\[ V_H \left( B, \frac{ B }{ 2 } \right) ~~=~~ \frac{B}{2\ln B} ~~\geq~~ \min_{B \in (1,\infty) } \left\{ \frac{B}{2\ln B} \right\} ~~=~~ \frac{ e }{ 2 } ~~>~~ \frac{ 5 }{ 6 \ln 2 } \ . \]
We proceed by considering two cases, depending on whether $B$ is rational or not.

\paragraph{Case 1: $B$ is irrational.} The next auxiliary claim, whose proof is provided in Appendix~\ref{app:proof_clm_LB_V_J_irrational}, establishes a lower bound on the density coefficient $V_J(B,\alpha)$ in the scenario of an irrational expansion factor.  

\begin{claim} \label{clm:LB_V_J_irrational}
$V_J(B, \alpha) \geq \frac{ 3(1-1/m) }{2 \ln B }$, when $B < \frac{ 41 }{ 20 }$ is irrational and $\alpha \notin \{ \frac{ B }{ q } : q \in \bbN \}$.
\end{claim}

We observe that this claim suffices to establish $V_J(B,\alpha) > \frac{ 5 }{ 6 \ln 2}$. Indeed, when $m \ge 3$, by recalling that $B < \frac{ 41 }{ 20 }$, we have $V_J(B,\alpha) > \frac{ 1 }{ \ln B } > \frac{ 1 }{ \ln (41/20) } > \frac{ 5 }{ 6 \ln 2}$. When $m=2$, since $B = 2^{1/k}$ is irrational by our case hypothesis, $k$ takes a value of at least $2$. As such, $B \leq \sqrt{2}$ and therefore $V_J(B, \alpha) > \frac{ 3 }{4 \ln B } \geq \frac{ 3 }{4 \ln \sqrt{2} } > \frac{ 5 }{ 6 \ln 2}$.

\paragraph{Case 2: $B$ is rational.} As previously noted, when $B = m^{1/k} > 1$ is rational, it must be an integer of value at least $2$. In parallel, we know that $B < \frac{ 41 }{ 20 }$, meaning that the only permissible expansion factor is $B = 2$ and that $m = 2^k$ for some integer $k \geq 1$. To simplify the upcoming discussion, we observe that any such configuration generates the exact same interleaved grid for any given $\alpha$, since its primary grid is ${\cal G}_{2^k,k} = \{ (2^k)^{p/k} \cdot T_{\min}^* : p \in \bbZ \} = \{ 2^p \cdot T_{\min}^* : p \in \bbZ \}$, which is invariant to $k$. Consequently, we assume without loss of generality that $m = 2$ and $k = 1$; for this setting, it remains to identify an optimal offset $\alpha \in (1, 2)$.

To ensure the holding cost coefficient $V_H(2, 1, \alpha)$ does not exceed our target threshold, we must have $\frac{1}{\ln 2} \cdot ( \alpha + \frac{2}{\alpha} - 2 ) \leq \frac{ 5 }{ 6 \ln 2 }$, which is equivalent to $\alpha \in [\frac{4}{3}, \frac{3}{2}]$. The next claim, whose proof appears in Appendix~\ref{app:proof_clm_LB_2_1}, shows that the density coefficient $V_J(2,1,\alpha)$ is lower-bounded by $\frac{ 5 }{ 6 \ln 2 }$, for any $\alpha$ in this interval.

\begin{claim} \label{clm:LB_2_1}
$V_J(2,1, \alpha) \geq \frac{ 5 }{ 6 \ln 2 }$, for every $\alpha \in [\frac{4}{3}, \frac{3}{2}]$.
\end{claim}

\subsection{Proof of Claim~\ref{clm:prop_interleaved_rounding_final}} \label{subsec:proof_clm_prop_interleaved_rounding_final}

\paragraph{Proof of items~\ref{item:final_prop_1} and~\ref{item:final_prop_2}.} Following the logic behind the proof of Claim~\ref{clm:prop_shifted}, we compute the expected values of the rounding operator $\myround^{\alpha}_{m,k, \Theta}(t)$ and its reciprocal by identifying their relationship to $t$ as a function of the random shift $\Theta$. By recycling our previous notation, let $B = m^{1/k}$ be the base expansion factor and let $\beta = \log_B(\alpha) \in [0,1)$ be the logarithmic offset. Moreover, let $L_t = \log_B(\frac{ t }{ T_{\min}^*} )$ be a mapping of $t > 0$ to the logarithmic scale. Once again, the rounded value $\myround^{\alpha}_{m,k, \Theta}(t)$ is given by the relative position of $L_t$ within a unit cycle, including an additional dependency on $\beta$, according to the next two cases:
\begin{enumerate}[label=(\Alph*), ref=(\Alph*)]
    \item \label{case:prop_interleaved_rounding_final_A} {\em When $L_t \leq \lfloor L_t \rfloor + \beta$}: Here, we have
    \[ \myround^{\alpha}_{m,k, \Theta}(t) ~~=~~ 
    \begin{cases} 
    t \cdot B^{\lfloor L_t \rfloor - L_t+ \beta + \Theta}, \qquad & \text{when } \Theta \in [0, L_t - \lfloor L_t \rfloor] \\
    t \cdot B^{\lfloor L_t \rfloor - L_t + \Theta}, & \text{when } \Theta \in (L_t - \lfloor L_t \rfloor, L_t - \lfloor L_t \rfloor + 1 - \beta] \\
    t \cdot B^{\lfloor L_t \rfloor - L_t -1 +\beta + \Theta}, & \text{when } \Theta \in (L_t - \lfloor L_t \rfloor + 1 - \beta,1]
    \end{cases} \]

    \item \label{case:prop_interleaved_rounding_final_B} {\em When $L_t > \lfloor L_t \rfloor + \beta$}: In this case,
    \[ \myround^{\alpha}_{m,k, \Theta}(t) ~~=~~ 
    \begin{cases} 
    t \cdot B^{\lfloor L_t \rfloor - L_t + 1 + \Theta}, \qquad & \text{when } \Theta \in [0, L_t - \lfloor L_t \rfloor - \beta] \\
    t \cdot B^{\lfloor L_t \rfloor - L_t + \beta + \Theta}, & \text{when } \Theta \in (L_t - \lfloor L_t \rfloor - \beta, L_t - \lfloor L_t \rfloor] \\
    t \cdot B^{\lfloor L_t \rfloor - L_t + \Theta}, & \text{when } \Theta \in (L_t - \lfloor L_t \rfloor, 1]
    \end{cases} \] 
\end{enumerate}

To determine the expected value of $\myround^\alpha_{m,k, \Theta}(t)$, we integrate these piecewise functions over $\Theta \in [0,1]$. Regardless of whether $L_t$ falls into case~\ref{case:prop_interleaved_rounding_final_A} or~\ref{case:prop_interleaved_rounding_final_B}, the resulting expression will be identical. We proceed by elaborating on case~\ref{case:prop_interleaved_rounding_final_A}, where
\begin{align*}
\exsub{ \Theta }{ \myround^\alpha_{m,k, \Theta}(t) } ~~=~~& \left( \int_{0}^{L_t - \lfloor L_t \rfloor} B^{\lfloor L_t \rfloor - L_t + \beta + \theta} \mathrm{d}\theta + \int_{L_t - \lfloor L_t \rfloor}^{L_t - \lfloor L_t \rfloor + 1 - \beta} B^{\lfloor L_t \rfloor - L_t + \theta} \mathrm{d}\theta \right. \\
& \qquad \qquad \left. \mbox{} + \int_{L_t - \lfloor L_t \rfloor + 1 - \beta}^{1} B^{\lfloor L_t \rfloor - L_t - 1 + \beta + \theta} \mathrm{d}\theta \right) \cdot t \\
~~=~~& \frac{1}{\ln B} \cdot \left( \left. B^{\lfloor L_t \rfloor - L_t + \beta + \theta} \right]_{0}^{L_t - \lfloor L_t \rfloor} + \left. B^{\lfloor L_t \rfloor - L_t + \theta} \right]_{L_t - \lfloor L_t \rfloor}^{L_t - \lfloor L_t \rfloor + 1 - \beta} \right. \\
& \qquad \qquad \left. \mbox{} + \left. B^{\lfloor L_t \rfloor - L_t - 1 + \beta + \theta} \right]_{L_t - \lfloor L_t \rfloor + 1 - \beta}^{1} \right) \cdot t \\
~~=~~ &\frac{B^{\beta} + B^{1-\beta} - 2}{\ln B} \cdot t \ .
\end{align*}
Substituting $B = m^{1/k}$ and $B^{\beta} = \alpha$ yields $\exsubpar{ \Theta }{ \myround^\alpha_{m,k, \Theta}(t) } = \frac{ k }{ \ln m } \cdot ( \alpha + \frac{ m^{1/k} }{\alpha} - 2 ) \cdot t$. Via an analogous calculation for case~\ref{case:prop_interleaved_rounding_final_B}, we arrive at the same result, concluding item~\ref{item:final_prop_1}.

For the reciprocal, by considering case~\ref{case:prop_interleaved_rounding_final_A}, it follows that
\begin{align*}
\exsub{ \Theta }{ \frac{1}{\myround^\alpha_{m,k, \Theta}(t)} } ~~=~~& \left( \int_{0}^{L_t - \lfloor L_t \rfloor} \frac{ \mathrm{d}\theta }{B^{\lfloor L_t \rfloor - L_t + \beta + \theta}}  + \int_{L_t - \lfloor L_t \rfloor}^{L_t - \lfloor L_t \rfloor + 1 - \beta} \frac{ \mathrm{d}\theta }{B^{\lfloor L_t \rfloor - L_t + \theta}}  \right. \\
& \qquad \qquad \left. \mbox{} + \int_{L_t - \lfloor L_t \rfloor + 1 - \beta}^{1} \frac{ \mathrm{d}\theta }{B^{\lfloor L_t \rfloor - L_t - 1 + \beta + \theta}}  \right) \cdot \frac{1}{t} \\
~~=~~& -\frac{1}{\ln B} \cdot \left( \left. B^{-(\lfloor L_t \rfloor - L_t + \beta + \theta)} \right]_{0}^{L_t - \lfloor L_t \rfloor} + \left. B^{-(\lfloor L_t \rfloor - L_t + \theta)} \right]_{L_t - \lfloor L_t \rfloor}^{L_t - \lfloor L_t \rfloor + 1 - \beta} \right. \\
& \qquad \qquad \left. \mbox{} + \left. B^{-(\lfloor L_t \rfloor - L_t - 1 + \beta + \theta)} \right]_{L_t - \lfloor L_t \rfloor + 1 - \beta}^{1} \right) \cdot \frac{1}{t} \\
~~=~~& \frac{2 - B^{-\beta} - B^{\beta - 1}}{\ln B} \cdot \frac{1}{t} \ .
\end{align*}
By plugging in $B = m^{1/k}$ and $B^{\beta} = \alpha$, we have $\exsubpar{ \Theta }{ \frac{ 1 }{ \myround^\alpha_{m,k, \Theta}(t) } } = \frac{ k }{ \ln m } \cdot ( 2 - \frac{1}{\alpha} - \frac{\alpha}{m^{1/k}} ) \cdot \frac{ 1 }{ t } $. An identical result is obtained for case~\ref{case:prop_interleaved_rounding_final_B}, thereby completing the proof of item~\ref{item:final_prop_2}.

\paragraph{Proof of item~\ref{item:final_prop_3}.} By specializing the notation of items~\ref{item:final_prop_1} and~\ref{item:final_prop_2} to $t = T_{\min}^*$, we observe that the anchor point $T_{\min}^*$ of our yet-unshifted grid ${\cal H}_{m,k}^{ \alpha }$ corresponds to a very specific setting, where $L_{ T_{\min}^* } = 0$. Consequently, the density coefficient ${\cal D}_{m,k,\Theta}^{ \alpha }$ can be expressed as a function of $\Theta$, depending on the shifted grid point to which $T_{\min}^*$ is $\myround^\alpha_{m,k, \Theta}$-rounded. The latter can be determined via the relation between $\Theta$ and the logarithmic offset $\beta = \log_B(\alpha)$ according to the following case disjunction: 
\[ {\cal D}_{m,k,\Theta}^{ \alpha } ~~=~~ 
\begin{cases} 
\frac{ 1 }{ T_{\min}^* \cdot B^\Theta } \cdot \Sigma_{m,k}^{ \alpha, \myprim}, \qquad & \text{when } \Theta \in [0, 1 - \beta] \\
\frac{ 1 }{ T_{\min}^* \cdot B^{\beta - 1 + \Theta} } \cdot \Sigma_{m,k}^{ \alpha, \mysecond}, & \text{when } \Theta \in (1 - \beta, 1]
\end{cases} \]
By integrating this functional dependency over $\Theta \in [0,1]$,
\begin{align*}
\exsub{ \Theta }{ {\cal D}_{m,k,\Theta}^{ \alpha } } ~~=~~& \left( \Sigma_{m,k}^{ \alpha, \myprim} \cdot \int_{0}^{1 - \beta} \frac{ \mathrm{d}\theta }{ B^{\theta} } + \Sigma_{m,k}^{ \alpha, \mysecond} \cdot \int_{1 - \beta}^{1} \frac{ \mathrm{d}\theta }{ B^{\beta - 1 + \theta} } \right) \cdot \frac{ 1 }{ T_{\min}^* } \\
~~=~~& \frac{ 1 }{ \ln B } \cdot \left( \Sigma_{m,k}^{ \alpha, \myprim} \cdot \left( 1 - B^{-(1-\beta)} \right) + \Sigma_{m,k}^{ \alpha, \mysecond} \cdot \left( 1 - B^{-\beta} \right) \right) \cdot \frac{ 1 }{ T_{\min}^* } \ .
\end{align*}
As before, substituting $B = m^{1/k}$ and $\beta = \log_B(\alpha)$ leads to the expression in item~\ref{item:final_prop_3}.
\section{Concluding Remarks} \label{sec:conclusions}

To conclude, we provide a broader perspective on our main findings and discuss promising avenues for future research. Specifically, we explore the extensibility of our algorithmic framework to upward-closed constraint structures, outline how shifted power-of-$m^{1/k}$ policies can be efficiently derandomized, and highlight fundamental questions regarding the approximability of resource-constrained joint replenishment.

\paragraph{Upward-closed constraints.} A key feature of the rounding methodology presented in Sections~\ref{sec:unshifted}--\ref{sec:Interleaved} is that its performance guarantees are not strictly tied to the linear-in-reciprocal structure of resource consumption. Our feasibility argument relies on making use of expansive rounding operators, which ensure that any upward-closed constraint remains satisfied, allowing extensions to variants of the resource-constrained joint replenishment problem whose feasible region of time intervals is upward-closed. The primary requirement for this extension is the ability to solve the convex relaxation~\eqref{eqn:conv_relax_RCJRP} to optimality in polynomial time. As long as the family of constraints in question allows for efficient optimization, capturing practical operational realities such as minimum order quantities or nested storage limitations via second-order cones or other well-behaved convex sets, all established approximation ratios remain valid. This degree of flexibility indicates that power-of-$m^{1/k}$ policies can be applied in a much wider array of production and distribution settings than initially described.

\paragraph{Derandomization.}
Even though most of our results are derived through randomized shifting, these policies can be efficiently derandomized to produce deterministic counterparts with identical performance guarantees. Indeed, since the rounded time interval of each commodity is proportional to $m^{\Theta/k}$, within any interval of $\Theta$ where the resulting grid points are fixed, the terms we utilize to upper-bound operational cost functions such as $F(T^{m,k,\alpha,\Theta})$ are proportional to $m^{-\Theta/k}$ and $m^{\Theta/k}$, representing ordering and holding costs, respectively. This structure implies that our bounds are piecewise-convex with a polynomial number of critical breakpoints where the rounding operator $\myround$ transitions to the next grid point. By identifying these $O(n k)$ values and computing the local minimum within each convex piece via standard calculus arguments, the best-possible shift can be determined in polynomial time. This approach ensures that the $\frac{5}{6\ln 2}$-approximation stated in Theorem~\ref{thm:main_result} is achievable without any reliance on stochasticity. Consequently, our approach maintains its practical utility for real-world inventory systems that require predictable ordering schedules.

\paragraph{APX-hardness.} While our work represents a substantial improvement on the classical approximation guarantee of $\frac{ 1 }{ \ln 2 } \approx 1.4427$, a captivating research direction would investigate whether resource-constrained joint replenishment is APX-hard or not. Namely, is there an absolute constant $\delta > 0$ for which it is NP-hard to approximate this problem within factor $1+\delta$? In contrast to the simpler case of $D=O(1)$ resources, recently proven to admit an $O( n^{ \tilde{O}( D^3 / \eps^4 ) } )$-time approximation scheme \citep{Segev25improved}, the computational status of arbitrarily many resources is still unresolved. Establishing either APX-hardness results or approximation schemes in this context would answer a long-standing open question in the field of lot-sizing and inventory coordination.

% --- Bibliography ---
\phantomsection
\addcontentsline{toc}{section}{References}
\bibliographystyle{plainnat}
\bibliography{BIB-JRP}

% --- Appendix ---
\addtocontents{toc}{\protect\setcounter{tocdepth}{2}} % Stop TOC listing for Appendix
\appendix
% MORE PROOFS FROM PROOF OF CONCEPT SECTION
\section{Additional Proofs}

\subsection{Proof of Claim~\ref{clm:optimal_base_unshift}} \label{app:proof_clm_optimal_base_unshift}

Let $V(m, k) = \max \{ V_J(m,k), V_H(m,k) \}$ be the objective function we wish to minimize, subject to $m \in \bbZ_{\geq 2}$ and $k \in \bbZ_{\geq 1}$, where $V_J(m, k) = ( 1 - \frac{1}{m} ) \cdot \frac{1}{m^{1/k}-1}$ and $V_H(m,k) = m^{1/k}$. To this end, standard calculus arguments show that $\min_{x \in \bbR_{++}} V(m,x) = \frac{ 1 }{ 2 }(1 + \sqrt{5 - \frac{ 4 }{ m }})$ for every $m \geq 1$, which is an increasing function of $m$. Moreover, $V(m, x)$ is decreasing in $x$ over the interval $(0,x^*_m]$ and increasing over $[x^*_m, \infty)$, where $x^*_m = \log_{\frac{1 + \sqrt{5 - 4/m}}{2}} (m)$ is the minimizer of this function. We proceed by considering two cases:
\begin{itemize}
    \item {\em When $m = 2$}: Here, $x^*_2 \approx 2.2223$, meaning that the candidate integer values for minimizing $V(2,k)$ are $k = 2$ and $k=3$. Since $V(2,2) < V(2,3)$, the optimal choice for $m=2$ is $k=2$.

    \item {\em When $m \geq 3$}: In this case, for every $k \in \bbZ_{\geq 1}$, 
    \[ V(m, k) ~~\geq~~ \min_{x \in \bbR_{++}} V(3,x) ~~=~~ \frac{ 1 }{ 2 } \left( 1 + \sqrt{5 - \frac{ 4 }{ 3 }} \right) ~~>~~ \sqrt{2} ~~=~~ V(2,2) \ . \]
    Therefore, there is no choice of $k$ outperforming the $(2,2)$ combination of the previous item. 
\end{itemize}

\subsection{Proof of Claim~\ref{clm:optimal_base_shifted}} \label{app:proof_clm_optimal_base_shifted}

Let $V(m, k) = \max \{ V_J(m, k), V_H(m, k) \}$ be the objective function we wish to minimize, subject to $m \in \bbZ_{\geq 2}$ and $k \in \bbZ_{\geq 1}$, where $V_J(m, k) = ( 1 - \frac{ 1 }{ m } ) \cdot \frac{ k }{ \ln m }$ and $V_H(m, k) = \frac{ k(m^{1/k}-1)}{ \ln m }$. One can verify that $V(3,2) = \frac{2(\sqrt{3}-1)}{\ln 3} < 1.3327$, and to prove that this value is the global minimum, we show that $V(m,k) \geq V(3,2)$ for any feasible integer pair.

To this end, letting $x = \frac{\ln m}{k} > 0$, we can rewrite our holding cost multiplier as 
\[ V_H(m, k) ~~=~~ \frac{ k(m^{1/k}-1)}{ \ln m } ~~=~~ \frac{ k(e^{\frac{\ln m}{k}}-1)}{ \ln m } ~~=~~ \frac{ e^x - 1}{ x } \ . \]
The function $\psi(x) = \frac{ e^x - 1}{ x }$ is strictly increasing over $(0,\infty)$, with $\psi( \frac{\ln 3}{2} ) = V(3,2)$. As such, for a pair of integers $m \in \bbZ_{\geq 2}$ and $k \in \bbZ_{\geq 1}$ to satisfy $V(m,k) < V(3,2)$, we must have $\frac{\ln m}{k} < \frac{\ln 3}{2}$. By applying this necessary condition to our joint ordering function, it follows that $V_J(m, k) = ( 1 - \frac{ 1 }{ m } ) \cdot \frac{ k }{ \ln m } > ( 1 - \frac{ 1 }{ m } ) \cdot \frac{2}{\ln 3}$. Given these observations, we proceed by considering three cases, depending on the value of $m$:
\begin{itemize}
    \item {\em When $m \geq 4$}: In this case, $V_J(m, k) > ( 1 - \frac{ 1 }{ 4 } ) \cdot \frac{2}{\ln 3} > 1.36 > V(3,2)$.

    \item {\em When $m = 3$}: Excluding the configuration $(3,2)$, for the boundary case of $k = 1$, we have $V_H(3,1) = \frac{2}{\ln 3} > 1.82 > V(3,2)$. Then, pairs of the form $(3,k)$ with $k \geq 3$ do not improve the objective value either, since $V_J(3,k) = \frac{2k}{3\ln 3} \geq \frac{2}{\ln 3} > 1.82 > V(3,2)$.

    \item {\em When $m = 2$}: Here, for $k = 1$, we get $V_H(2,1) = \frac{1}{\ln 2} > 1.44 > V(3,2)$. Similarly, for $k \geq 2$, our joint ordering multiplier becomes $V_J(2,k) = \frac{k}{2\ln 2} \geq \frac{1}{\ln 2} > 1.44 > V(3,2)$.
\end{itemize}

\subsection{Proof of Claim~\ref{clm:LB_V_J_irrational}} \label{app:proof_clm_LB_V_J_irrational}

Let us observe that, in terms of $B$, the expected density coefficient can be written as 
\begin{equation} \label{eqn:clm_LB_V_J_irrational_1}
V_J(B,\alpha) ~~=~~ \frac{ 1 }{ \ln B } \cdot \left( \left( 1 - \frac{\alpha}{ B } \right) \cdot \Sigma_{B}^{ \alpha, \myprim} + \left( 1 - \frac{1}{\alpha} \right) \cdot \Sigma_{B}^{ \alpha, \mysecond} \right) \ ,
\end{equation}
where $\Sigma_{B}^{ \alpha, \myprim} = \sum_{\emptyset \neq {\cal N} \subseteq {\cal S}_{B}^{ \alpha, \myprim}} \frac{ (-1)^{ |{\cal N}| + 1} }{ \mylcm({\cal N}) }$ and $\Sigma_{B}^{ \alpha, \mysecond} = \sum_{\emptyset \neq {\cal N} \subseteq {\cal S}_{B}^{ \alpha, \mysecond}} \frac{ (-1)^{ |{\cal N}| + 1} }{ \mylcm({\cal N}) }$. To evaluate these inclusion-exclusion summations, we partition the domain of $\alpha \in (1, B)$ into three cases, based on its commensurability with integer powers of $B$.

\paragraph{Case 1: $\alpha$ is rational.} Let us write $\alpha = \frac{ p }{ q }$ as a rational fraction in lowest terms, with $p > q \geq 1$. In regard to $\Sigma_{B}^{ \alpha, \myprim}$, by recalling that ${\cal S}_B^{ \alpha, \myprim } = \{ B^{\kappa-1} \}_{\kappa \in [k]} \cup \{ \alpha B^{\kappa-1} \}_{\kappa \in [k]}$ and that $B$ is irrational, for every subset $\emptyset \neq {\cal N} \subseteq {\cal S}_{B}^{ \alpha, \myprim}$ we have $\mylcm( {\cal N} ) < \infty$ only when ${\cal N}$ is either a singleton or a cross-grid pair of the form ${\cal N} = \{ B^{\kappa-1}, \alpha B^{\kappa-1} \}$, for some $\kappa \in [k]$. Therefore, letting $\Sigma_k = \sum_{\kappa \in [k]} \frac{ 1 }{ m^{(\kappa-1)/k} } = \frac{ 1 - 1/m }{ 1 - 1/B }$, it follows that $\Sigma_{B}^{ \alpha, \myprim} = (1 + \frac{1}{\alpha} - \frac{1}{ \mylcm(1,\alpha)}) \cdot \Sigma_k = (1 + \frac{q}{p} - \frac{1}{ p}) \cdot \Sigma_k$. By contrast, in relation to $\Sigma_{B}^{ \alpha, \mysecond} $, one has ${\cal S}_B^{ \alpha, \mysecond } = \{ B^{\kappa-1} \}_{\kappa \in [k]} \cup \{ \frac{1}{\alpha} B^{\kappa} \}_{\kappa \in [k]}$. Crucially, since $B$ is irrational, for every subset $\emptyset \neq {\cal N} \subseteq {\cal S}_B^{ \alpha, \mysecond }$, we have $\mylcm( {\cal N} ) < \infty$ only when ${\cal N}$ is a singleton. Indeed, since $\{ \frac{1}{\alpha} B^{\kappa} \}_{\kappa \in [k]}$ is a $\frac{ B }{\alpha}$-scaling of $\{ B^{\kappa-1} \}_{\kappa \in [k]}$, with $\frac{ B }{\alpha}$ being irrational, all cross-grid subsets ${\cal N}$ have $\mylcm({\cal N}) =\infty$. As a result, $\Sigma_{B}^{ \alpha, \mysecond} = (1 + \frac{\alpha}{B}) \cdot \Sigma_k = (1 + \frac{p}{qB} ) \cdot \Sigma_k$. By substituting these terms into~\eqref{eqn:clm_LB_V_J_irrational_1},
\begin{align*} 
V_J(B,\alpha) ~~=~~ & \frac{ 1 }{ \ln B } \cdot \left( \left( 1 - \frac{\alpha}{ B } \right) \cdot \left( 1 + \frac{q}{p} - \frac{ 1 }{ p } \right) + \left( 1 - \frac{1}{\alpha} \right) \cdot \left( 1 + \frac{p}{qB} \right) \right) \cdot \Sigma_k \\
~~=~~ & \frac{ 1 }{ \ln B } \cdot \left( 2 \cdot \left( 1 - \frac{1}{B} \right) - \frac{qB - p}{pqB} \right) \cdot \Sigma_k \\
~~=~~ & \frac{ 1-1/m }{ \ln B } \cdot \left( 2 - \frac{B - p/q}{p(B-1)} \right) \\
~~\geq~~ & \frac{ 3(1-1/m) }{ 2\ln B } \ ,
\end{align*}
where the last inequality holds since $\alpha = \frac{ p }{ q } \in (1, B)$, implying that $\frac{ p }{ q } > 1$ and $p \geq 2$.

\paragraph{Case 2: $\alpha$ is irrational, and $\mylcm( \alpha, B^{\kappa}) = \infty$ for every $\kappa \in \bbZ$.} In this case, the ratio of any element in the secondary grid ${\cal G}_B^{\alpha}$ to any element in the primary grid ${\cal G}_B$ is irrational, implying that the cross-grid $\mylcm$ intersections vanish in both $\Sigma_{B}^{ \alpha, \myprim}$ and $\Sigma_{B}^{ \alpha, \mysecond}$. As such, $\Sigma_{B}^{ \alpha, \myprim} = (1 + \frac{1}{\alpha}) \cdot \Sigma_k$ and $\Sigma_{B}^{ \alpha, \mysecond} = (1 + \frac{\alpha}{B}) \cdot \Sigma_k$. By substituting these terms into~\eqref{eqn:clm_LB_V_J_irrational_1},
\begin{align*} 
V_J(B,\alpha) ~~=~~ & \frac{ 1 }{ \ln B } \cdot \left( \left( 1 - \frac{\alpha}{ B } \right) \cdot \left( 1 + \frac{1}{\alpha} \right) + \left( 1 - \frac{1}{\alpha} \right) \cdot \left( 1 + \frac{\alpha}{B} \right) \right) \cdot \Sigma_k \\
~~=~~ & \frac{ 2 }{ \ln B } \cdot \left( 1 - \frac{1}{ B } \right) \cdot \Sigma_k \\
~~=~~ & \frac{ 2(1-1/m) }{ \ln B } \ .
\end{align*}

\paragraph{Case 3: $\alpha$ is irrational, and $\mylcm( \alpha, B^{\kappa}) < \infty$ for some $\kappa \in \bbZ$.} Here, since $B^k = m$ is an integer, there exist integers $p$, $q$, and $1 \leq \kappa \leq k-1$ for which $\alpha = \frac{ p }{ q } \cdot B^{\kappa}$. Without loss of generality, we assume that $p$ and $q$ are coprime. Additionally, the maximal overlap between the primary and secondary grids occurs when $\kappa = 1$, and we therefore lower-bound the worst-case scenario, where $\alpha = \frac{ p }{ q } \cdot B$. Specifically, by employing the second-order Bonferroni inequality, $\Sigma_{B}^{ \alpha, \myprim} \geq (1 + \frac{1}{\alpha} - \frac{1}{ \mylcm(1,\alpha/B)}) \cdot \Sigma_k = (1 + \frac{q}{pB} - \frac{1}{ p}) \cdot \Sigma_k$, and similarly, $\Sigma_{B}^{ \alpha, \mysecond} \geq (1 + \frac{p}{q} - \frac{1}{ q}) \cdot \Sigma_k$. By substituting these terms into~\eqref{eqn:clm_LB_V_J_irrational_1}, 
\begin{align*} 
V_J(B,\alpha) ~~\geq~~ & \frac{ 1 }{ \ln B } \cdot \left( \left( 1 - \frac{p}{ q } \right) \cdot \left( 1 + \frac{q}{pB} - \frac{ 1 }{ p } \right) + \left( 1 - \frac{q}{pB} \right) \cdot \left( 1 + \frac{p}{q} - \frac{ 1 }{ q } \right) \right) \cdot \Sigma_k \\
~~=~~ & \frac{ 1 }{ \ln B } \cdot \left( 1 - \frac{1}{B} \right) \cdot \left( 2 - \frac{1}{p} \right) \cdot \Sigma_k \\
~~=~~ & \frac{ 1-1/m }{ \ln B } \cdot \left( 2 - \frac{1}{p} \right) \\
~~\geq~~ & \frac{ 3(1-1/m) }{ 2\ln B }\ ,
\end{align*}
where the last inequality holds since $p \geq 2$. Otherwise, $p=1$, meaning that $\alpha = \frac{ B }{ q }$ for some integer $q$, which contradicts the hypothesis of Claim~\ref{clm:LB_V_J_irrational}, stating that $\alpha \notin \{ \frac{ B }{ q } : q \in \bbN \}$.

\subsection{Proof of Claim~\ref{clm:LB_2_1}} \label{app:proof_clm_LB_2_1}

Within the interval $[\frac{4}{3}, \frac{3}{2}]$, the right endpoint $\alpha = \frac{ 3 }{ 2 }$ takes us back to our original configuration, for which we already know that $V_J(2,1,\frac{3}{2}) = \frac{ 5 }{ 6 \ln 2}$. Additionally, the left endpoint $\alpha = \frac{ 4 }{ 3 }$ produces an identical objective value, since  $V_J(2,1,\frac{4}{3}) = V_H(2,1,\frac{4}{3}) = \frac{ 5 }{ 6 \ln 2}$. In what follows, we show that all interior points $\alpha \in (\frac{4}{3}, \frac{3}{2})$ do not improve on $\frac{ 5 }{ 6 \ln 2}$. To this end, by exploiting the notation introduced in Appendix~\ref{app:proof_clm_LB_V_J_irrational}, our analysis will evaluate the inclusion-exclusion summations $\Sigma_2^{ \alpha, \myprim} = \sum_{\emptyset \neq {\cal N} \subseteq {\cal S}_2^{ \alpha, \myprim}} \frac{ (-1)^{ |{\cal N}| + 1} }{ \mylcm({\cal N}) }$ and $\Sigma_2^{ \alpha, \mysecond} = \sum_{\emptyset \neq {\cal N} \subseteq {\cal S}_2^{ \alpha, \mysecond}} \frac{ (-1)^{ |{\cal N}| + 1} }{ \mylcm({\cal N}) }$ in two distinct cases, depending on whether $\alpha$ is rational or not. 

\paragraph{Case 1: $\alpha$ is irrational.} In this case, the ratio of any element in the secondary grid ${\cal G}_2^{\alpha} = \{ {\alpha} 2^p \cdot T_{\min}^* : p \in \bbZ \}$ to any element in the primary grid ${\cal G}_2 = \{ 2^p \cdot T_{\min}^* : p \in \bbZ \}$ is irrational, implying that the cross-grid $\mylcm$ intersections vanish in both $\Sigma_2^{ \alpha, \myprim}$ and $\Sigma_2^{ \alpha, \mysecond}$. As such, $\Sigma_2^{ \alpha, \myprim} = 1 + \frac{1}{\alpha}$ and $\Sigma_2^{ \alpha, \mysecond} = 1 + \frac{\alpha}{2}$. By substituting these terms into~\eqref{eqn:clm_LB_V_J_irrational_1},
\[ V_J(2,1,\alpha) ~~=~~ \frac{ 1 }{ \ln 2 } \cdot \left( \left( 1 - \frac{\alpha}{ 2 } \right) \cdot \left( 1 + \frac{1}{\alpha} \right) + \left( 1 - \frac{1}{\alpha} \right) \cdot \left( 1 + \frac{\alpha}{2} \right) \right) ~~=~~ \frac{ 1 }{ \ln 2 } ~~\geq~~ \frac{ 5 }{ 6\ln 2} \ . \]

\paragraph{Case 2: $\alpha$ is rational.} Let us write this offset as a reduced fraction $\alpha = \frac{p}{q}$, where $p$ and $q$ are coprime integers. As such, the primary and secondary base multipliers for the interleaved grid can be expressed as ${\cal S}_2^{ \alpha, \myprim} = \{1, \frac{p}{q}\}$ and ${\cal S}_2^{ \alpha, \mysecond} = \{1, \frac{2q}{p}\}$, implying that $\Sigma_2^{ \alpha, \myprim} = 1 + \frac{q}{p} - \frac{ 1 }{ \mylcm(1, p/q) } = 1 + \frac{q}{p} - \frac{ 1 }{ p }$ and $\Sigma_2^{ \alpha, \mysecond} = 1 + \frac{p}{2q} - \frac{ 1 }{ \mylcm(1, 2q/p) }$. Since the latter LCM-term depends on the parity of $p$, we proceed by examining two subcases.

\paragraph{Subcase 2A: $p$ is odd.} In this case, $\mylcm(1, 2q/p) = 2q$, and by substituting $\Sigma_2^{ \alpha, \myprim}$ and $\Sigma_2^{ \alpha, \mysecond}$ into~\eqref{eqn:clm_LB_V_J_irrational_1},
\begin{align*}
V_J \left( 2,1,\frac{p}{q} \right) ~~=~~ & \frac{ 1 }{ \ln 2 } \cdot \left( \left( 1 - \frac{p}{ 2q } \right) \cdot \left( 1 + \frac{q}{p} - \frac{ 1 }{ p} \right) + \left( 1 - \frac{q}{p} \right) \cdot \left( 1 + \frac{p}{2q} - \frac{ 1 }{ 2q} \right) \right) \\
~~=~~ & \frac{ 1 }{ \ln 2 } \cdot \left( 1 - \frac{1}{2p} \right) \\
~~\geq~~ & \frac{ 5 }{ 6\ln 2 } \ .
\end{align*}
Here, the last inequality follows by observing that $p \geq 3$. Indeed, this parameter is odd, and we cannot have $p=1$ or otherwise $\alpha = \frac{ 1 }{ q } \notin (\frac{4}{3}, \frac{3}{2})$.

\paragraph{Subcase 2B: $p$ is even.} Suppose that $p = 2r$. Then, since $p$ and $q$ are coprime, $r$ and $q$ must be coprime as well, meaning that $\mylcm(1, 2q/p) = \mylcm(1, q/r) = q$. By substituting $\Sigma_2^{ \alpha, \myprim}$ and $\Sigma_2^{ \alpha, \mysecond}$ into~\eqref{eqn:clm_LB_V_J_irrational_1},
\begin{align*}
V_J \left( 2,1,\frac{p}{q} \right) ~~=~~ & \frac{ 1 }{ \ln 2 } \cdot \left( \left( 1 - \frac{p}{ 2q } \right) \cdot \left( 1 + \frac{q}{p} - \frac{ 1 }{ p} \right) + \left( 1 - \frac{q}{p} \right) \cdot \left( 1 + \frac{p}{2q} - \frac{ 1 }{ q} \right) \right) \\
~~=~~ & \frac{ 1 }{ \ln 2 } \cdot \left( 1 - \frac{1}{2q} \right) \\
~~\geq~~ & \frac{ 5 }{ 6\ln 2 } \ ,
\end{align*}
where the last inequality holds since $q \geq 3$. To verify this claim, note that since $p$ and $q$ are coprime, and since $p$ is even by the subcase hypothesis, $q$ must be odd. In addition, we cannot have $q = 1$ or otherwise $\alpha = p \notin (\frac{4}{3}, \frac{3}{2})$.

\end{document}